\documentclass[a4paper]{article}
\pdfoutput=1
\usepackage{faupub}
\usepackage[utf8]{inputenc}

\usepackage[table ]{ xcolor}
\usepackage{subcaption}
\usepackage{multirow}
\usepackage{comment}
\usepackage{amsmath}

\usepackage{tikz}
\usetikzlibrary{shapes,arrows}
\tikzstyle{startstop} = [rectangle,rounded corners, minimum width=3cm,minimum height=1cm,text centered, draw=black,fill=red!30]
\tikzstyle{io} = [trapezium, trapezium left angle = 70,trapezium right angle=110,minimum width=3cm,minimum height=1cm,text centered,draw=black,fill=blue!30]
\tikzstyle{process} = [rectangle,minimum width=3cm,minimum height=1cm,text centered,text width =3cm,draw=black,fill=orange!30]
\tikzstyle{decision} = [diamond,minimum width=3cm,minimum height=1cm,shape aspect=3,inner sep = 0.4pt,text centered,draw=black,fill=green!30]
\tikzstyle{arrow} = [thick,->,>=stealth]
\tikzstyle{shadow}=[preaction={fill=black,opacity=.5,transform canvas={xshift=0.5mm,yshift=-0.5mm},shading=radial,shading angle=20},fill=red]

\tikzstyle{ellipse}=[draw, rectangle, minimum width=2.8em, rounded corners=6pt,line width=0.5pt]
\tikzstyle{pxsbx}=[trapezium, trapezium left angle=75, trapezium right angle=105, minimum width=3em, text centered, draw = black, fill=white,line width=0.5pt] 
\tikzstyle{lingxing}=[draw,diamond,shape aspect=3,inner sep = 0.4pt,thick,font=\itshape,line width=0.5pt]

\usepackage{amssymb}
\usepackage{graphicx,color}

\usepackage{amsmath,amsthm}

\usepackage{mathrsfs}

\usepackage{bm}

\ifx\Ref\undefined
\newcommand{\Ref}[1]{(\ref{#1})}
\fi

\newtheorem{Theorem}{Theorem}[section]
\newtheorem{Definition}{Definition}[section]
\newtheorem{Lemma}[Theorem]{Lemma}
\newtheorem{Corollary}{Corollary}[section]
\newtheorem{Proposition}[Theorem]{Proposition}
\theoremstyle{remark}
\newtheorem{Remark}{Remark}[section]

\newcommand{\R}{\mathbb{R}}
\newcommand{\C}{\mathbb{C}}

\theoremstyle{remark}

\newcommand{\ccirc}{\kern0.2ex\vcenter{\hbox{$\scriptstyle\circ$}}\kern0.2ex}

\newcommand{\SLDC}{{\mathrm{SL}(2,\mathbb{C})}}

\newcommand{\Slc}{\mathrm{SL}(2,\mathbb{C})}

\newcommand{\Su}{\mathrm{SU}(2)}

\newcommand{\ransu}{\rangle_{\rm SU(1,1)}}
\newcommand{\lansu}{\langle}
\newcommand{\raninv}{\rangle_{\rm Inv}}
\newcommand{\laninv}{\langle}

\def\be{\begin{eqnarray}}
\def\ee{\end{eqnarray}}

\newcommand{\ca}{\mathcal A}

\newcommand{\cc}{\mathcal C}
\newcommand{\cd}{\mathcal D}

\newcommand{\cg}{\mathcal G}
\newcommand{\ch}{\mathcal H}
\newcommand{\ci}{\mathcal I}

\newcommand{\ck}{\mathcal K}

\newcommand{\cn}{\mathcal N}

\newcommand{\cz}{\mathcal Z}

\newcommand{\fl}{\mathfrak{l}}  
  
\newcommand{\fn}{\mathfrak{n}}

  \newcommand{\Fu}{\mathfrak{U}}

\renewcommand{\b}{\beta}
\newcommand{\g}{\gamma}
\newcommand{\G}{\Gamma}

\newcommand{\eps}{\varepsilon}

\newcommand{\sig}{\sigma}

\renewcommand{\l}{\lambda}

\renewcommand{\O}{\Omega}

\newcommand{\rmd}{\mathrm d}

\newcommand{\lt}{\left}
\newcommand{\rt}{\right}

\newcommand{\lag}{\left\langle}
\newcommand{\rag}{\right\rangle}

\newcommand{\act}{\rhd}

\newcommand{\sgn}{\mathrm{sgn}}

\newcommand{\zz}{{\mathbf{z}}}
\newcommand{\CP}{\mathbb{CP}}
\newcommand{\iu}{i}

\newcommand{\Suo}{\text{SU}(1,1)}

\title{Finiteness of spinfoam vertex amplitude with timelike polyhedra, and the full amplitude}

\author[1,3]{Muxin Han}  

\author[2]{\ Wojciech Kaminski}  

\author[3]{\ Hongguang Liu}

\affiliation[1]{Department of Physics, Florida Atlantic University, 777 Glades Road, Boca Raton, FL 33431-0991, USA}

\affiliation[2]{Faculty of Physics, University of Warsaw, Pasteura 5, 02-093 Warsaw, Poland}

\affiliation[3]{Institut f\"ur Quantengravitation, Universit\"at Erlangen-N\"urnberg, Staudtstr. 7/B2, 91058 Erlangen, Germany}

\emailAdd{hanm(At)fau.edu}
\emailAdd{Wojciech.Kaminski(AT)fuw.edu.pl}
\emailAdd{hongguang.liu(At)gravity.fau.de}

\abstract{This work focuses on Conrady-Hnybida's 4-dimensional extended spinfoam model with timelike polyhedra, while we restrict all faces to be spacelike. Firstly, we prove the absolute convergence of the vertex amplitude with timelike polyhedra, when SU(1,1) boundary states are coherent states or the canonical basis, or their finite linear combinations. Secondly, based on the finite vertex amplitude and a proper prescription of the SU(1,1) intertwiner space, we construct the extended spinfoam amplitude on arbitrary cellular complex, taking into account the sum over SU(1,1) intertwiners of internal timelike polyhedra. We observe that the sum over SU(1,1) intertwiners is infinite for the internal timelike polyhedron that has at least 2 future-pointing and 2 past-pointing face-normals. In order to regularize the possible divergence from summing over SU(1,1) intertwiners, we develop a quantum cut-off scheme based on the eigenvalue of the ``shadow operator''. The spinfoam amplitude with timelike internal polyhedra (and spacelike faces) is finite, when 2 types of cut-offs are imposed: one is imposed on $j$ the eigenvalue of area operator, the other is imposed on the eigenvalue of shadow operator for every internal timelike polyhedron that has at least 2 future-pointing and 2 past-pointing face-normals. }


\begin{document}

\maketitle

\section{Introduction}

The spinfoam formulation of Loop Quantum Gravity (LQG) is a covariant approach toward the background-independent and nonperturbative quantum theory of gravity in 4 dimensions \cite{book1,book,rovelli2014covariant,Perez2012}. In this formulation, the spacetime geometry is discretized and quantized with the spinfoam amplitude, which generalizes the concept of Feynman path integral to LQG. The spinfoam amplitude describes the dynamics of LQG on a triangulation of 4-dimensional manifold. Among various models of the spinfoam, a popular model of the Lorentizian theory is constructed by Engle, Pereira, Rovelli, and Livine (EPRL) \cite{EPRL}, then there is an extension by Conrady and Hnybida to include timelike tetrahedra and faces \cite{Conrady:2010kc}, whereas tetrahedra and faces are all spacelike in the original EPRL model\footnote{The terminology of time/space-like tetrahedra comes from the ways of solving simplicity constraints with either the time-gauge $n_I=(1,0,0,0)$ or the space-gauge $n_I=(0,0,0,1)$ for tetrahedron normals $n_I$. These gauge fixings reduce the Lorentz group to either SU(2) or SU(1,1).}. Both the EPRL model and the Conrady-Hnybida extension can be generalized to arbitrary 4-dimensional cellular complex with polyhedra replacing tetrahedra \cite{KKL,generalize}.

In this paper, we refer to the model including timelike polyhedra and faces as the \emph{extended spinfoam model}. Recently, some important progress was made on the semiclassical analysis of the extended spinfoam model \cite{Kaminski:2017eew,Liu:2018gfc,Simao:2021qno}. The extended spinfoam model has the advantage of demonstrating better semiclassical behavior than the EPRL model. As an remarkable property, when both timelike and spacelike tetrahedra are present in a 4-simplex, both the vector geometry (degenerate 4-simplex) and Euclidean 4-simplex geometry are absent in the large-j asymptotics of the Lorentzian vertex amplitude \cite{Liu:2018gfc}. It is in contrast to the Lorentzian EPRL model where these 2 spurious types of geometries appear in the asymptotics.

Although the recent results demonstrates the promising aspects of the extended spinfoam model, there has been an important gap in the literature: There has been no proof of the finiteness of the vertex amplitude in the presence of timelike tetrahedra or polyhedra. The finiteness of spinfoam vertex amplitude is nontrivial since the gauge group $\Slc$ is non-compact. All proofs of the finiteness covered only the Barrett-Crane (BC) \cite{Baez:2001fh} and EPRL vertex amplitudes with only spacelike polyhedra \cite{Engle:2008ev,Kaminski:2010qb}, possibly with the $q$-deformations \cite{NP,QSF,QSF1,Han:2021tzw}. 

The purpose of the present paper is to partially fill this gap: In this work, we are taking into account the generalized spinfoam vertex dual to 4-cell with boundary polyhedra. We prove the finiteness of the vertex amplitude in presence of timelike polyhedra while all faces are still spacelike. According to \cite{Conrady:2010kc}, the spacelike face in the timelike polyhedron associates with the boundary state in the discrete-series unitary irreducible representations of SU(1,1), in contrast to SU(2) boundary states of spacelike polyhedra. The finiteness of the vertex amplitude holds when the SU(1,1) boundary states are coherent states or the elements of the canonical basis, or their finite linear combinations. Note that the finiteness of the vertex amplitude might depend on the choice of the state because the unitary irreducible representations of SU(1,1) are infinite-dimensional. Our result shows that the vertex amplitude is densely defined on the boundary Hilbert space. Although we do not consider timelike faces, our work is a crucial step because the vertex amplitude with timelike tetrahedra and spacelike triangles are already sufficient to ensure the absence of vector and Euclidean geometries in the large-$j$ asymptotics, from the vertex amplitude with boundary coherent states. The study of the vertex amplitude with timelike faces is postponed to a future publication, since the strategy of analysis will be completely different from this work.


We generalize our study of the finiteness to the spinfoam amplitude on any 4-dimensional cellular complex. Gluing vertex amplitudes introduces the sums over intertwiners associated to internal polyhedra. In the case of the EPRL amplitude with only spacelike internal polyhedra, the finiteness of the amplitude on complexes with fixed $j$ follows from the finiteness of the vertex amplitude since the intertwiner space of SU(2) is finite-dimensional. However this becomes nontrivial for amplitudes with timelike internal polyhedra since the SU(1,1) intertwiner space can be infinite-dimensional. In our discussion, all face-normals of timelike polyhedra are timelike since all faces are spacelike. We show that in the case when all timelike internal polyhedra have only 1 face-normal future-pointing or only 1 face-normal past-pointing\footnote{The face with future-pointing or past-pointing normal associates with states in the SU(1,1) irreduciable representation $D^+_j$ or $D^-_j$.}, their SU(1,1) intertwiner spaces are finite-dimensional, so the spinfoam amplitude with fixed $j$ is finite. However, in the difficult case that the timelike internal polyhedron has at least 2 future-pointing and 2 past-pointing face-normals, the intertwiner space is infinite-dimensional.  Geometrically, this relates to the non-compactness of the space of shapes of timelike polyhedron. In this case, we are not able to prove the finiteness of the amplitude, and unfortunately we do not have a definite answer about whether the spinfoam amplitude is finite or divergent in general.

We suspect that the amplitude in general situation is divergent due to the infinite sum of SU(1,1) intertwiners. In order to regulate the possible divergence from summing over SU(1,1) intertwiners, we develop a quantum cut-off scheme based on the eigenvalue of the ``shadow operator''. Classically we call a shadow the extremal area of a polytope as seen from any direction. The shadow operator promotes this notion to the quantum level. We show that the amplitude is finite once an upper bound is imposed on the eigenvalue of the shadow operator. 

{Our result is: The extended spinfoam amplitude on any cellular complex with spacelike and timelike internal polyhedra (with spacelike faces) is finite, when 2 types of cut-offs are imposed: one is imposed on $j$ the eigenvalue of area operator (in the case of the bubble divergence), the other is imposed on the eigenvalue of shadow operator for every internal timelike polyhedron that has at least 2 future-pointing and 2 past-pointing face-normals. }

This paper is organized as follows: In Section \ref{sec2}, we review the integral expression of the extended vertex amplitude $A_\G$ with the SU(2) and SU(1,1) boundary states. In Section \ref{sec3}, we prove that the $A_\G$ is absolutely convergent. In Section \ref{Gluing vertex amplitudes}, we discuss the SU(1,1) intertwiner space and  gluing vertex amplitudes at timelike polyhedra. In Section \ref{spinfoam amplitude with additional cut-off}, we discuss the situation of the internal polyhedra that has at least 2 future-pointing and 2 past-pointing face-normals, and introduce the shadow operator and the cut-off of its eigenvalue. Finally we conclude and discuss a few future perspectives in Section \ref{sec4}.

\section{Vertex amplitude}\label{sec2}

\subsection{Extended spinfoam model}

We introduce in detail our notations and definition of the model that we use in the paper. The unitary irreducible representations $\ch_{(\rho,n)}$ of $\SLDC$ from principal series can be described as a space of measureable functions
\be
\Psi(\zz),\quad \zz=\left(\begin{array}{l}
z_{+}\\
z_{-}
\end{array}\right)\in \C^2,
\ee
satisfying for almost every $\zz$ (with respect to Lebegue measure)
\be
\forall_{r\not=0,\phi}
\Psi(re^{i\phi}\zz)=r^{i\rho-2}e^{in\phi}\Psi(\zz),\label{cond-n-rho}
\ee
Following \cite{gelfand5} we denote variables by $\zz$ also for functions that are not holomorphic.
The action of $\SLDC$ is defined by
\be
g\act \Psi(\zz)=\Psi(g^T\zz),\quad g\in \SLDC
\ee
For two such functions $\Psi_1,\Psi_2\in\ch_{(\rho,n)}$,
a form
\be
\overline{\Psi_{1}(\mathbf{z})} \Psi_{2}(\mathbf{z}) \Omega_{\mathbf{z}},\quad
\Omega_{\mathbf{z}}=\frac{\mathrm{i}}{2}\left(z_{+}dz_{-}-z_{-}dz_{+}\right)\wedge\left(\bar{z}_{+}d\bar{z}_{-}-\bar{z}_{-}d\bar{z}_{+}\right)
\ee
descends to $\CP^1$, which is the quotient of $\C^2\setminus\{0\}$. We define a scalar product as an integral on $\CP^1$
\be
\left\langle\Psi_{1}, \Psi_{2}\right\rangle=\int_{\mathbb{C} \mathbb{P}^{1}} \overline{\Psi_{1}(\mathbf{z})} \Psi_{2}(\mathbf{z}) \Omega_{\mathbf{z}},
\ee
As mentioned above, the integral which is a priori defined in $\C^2$ descends in fact to the quotient space $\CP^1$.

The Hilbert space $\ch_{(\rho,n)}$ is defined as a space of those functions satisfying \eqref{cond-n-rho} with  finite norm
\be
\ch_{(\rho,n)}=\{\Psi\colon \langle \Psi,\,\Psi\rangle<\infty\}
\ee
We consider now $\rho=\gamma n$, $n=2j$. 

We denote by $D^\pm_j$ and $\cc^\chi_s$ the discrete and continuous series unitary irreducible representations of SU(1,1), and denote by $D^0_j$ the unitary irreducible representation of SU(2). Here we do not discuss the continuous series representation of SU(1,1). The spinfoam embedding map
$Y_\eps$ ($\eps=0,\pm$) is the spinfoam embedding map from the SU(1,1) or SU(2) unitary irreducible representations (in short irreps) into the $\Slc$ unitary irreducible representations\footnote{$\chi$ labels different continuous-series representations relating to integer or half-integer $j$.}
\be
Y_\pm:&&D_j^\pm\to\mathcal{H}_{(2\g j,2j)}\simeq\left(\bigoplus_{k>1/2}^{j}D_{k}^{+}\oplus\int_{s\geq 0}^{\oplus}\mathrm{d}s\,\mathcal{C}_{s}^\chi\right)\oplus\left(\bigoplus_{k>1/2}^j D_{k}^{-}\oplus\int_{s\geq0}^{\oplus}\mathrm{d}s\,\mathcal{C}^\chi_{s}\right), \label{ymap}\\
Y_0:&& D^0_j\to\mathcal{H}_{(2\g j,2j)}\simeq\bigoplus_{k\geq j} D^0_k
\ee
where $Y_\pm$ identifies $D_j^\pm$ to the subspace $D_{k=j}^\pm\subset\mathcal{H}_{(2\g j,2j)}$, and $Y_0$ identifies $D^0_j$ to the lowest subspace $D^0_{k=j}\subset \mathcal{H}_{(2\g j,2j)}$. We have $j>1/2$ in the case of SU(1,1) states. 

\subsection{Coherent states}

The Lie algebra of $\Slc$ has generators
\be
L^i=\sigma^i/2,\quad K^i=i\sigma^i/2,\quad i=1,2,3
\ee
where $\sig^i$ are Pauli matrices. SU(1,1) is the subgroup of $\Slc$ generated by $L^3,K^1,K^2$. The other subgroup SU(2) is generated by $L^1,L^2,L^3$. The SU(1,1) discrete-series unitary irrep $D_j^\pm$ has an orthonormal basis $|j,m\rangle^\pm$ of eigenfunctions of $L_3$
\be
|j,m\rangle^+\in  D_j^+,\quad m\geq j,\quad\text{and}\quad |j,m\rangle^-\in D_j^-,\quad m\leq -j,
\ee
similarly to an orthogonal basis in $\Su$ irrep
\be
|j,m\rangle^0\in D^0_j,\quad -j\leq m\leq j.
\ee
We refer to $|j,m\rangle^\eps$ ($\eps=0,\pm$) as the canonical basis.

The $Y_0$ map acting on $|j,m\rangle^0$ is equal to 
\begin{equation}\label{f_su2}
    F_{jm}(\mathbf{z})=\sqrt{\frac{\Gamma(2j+2)}{2\pi\Gamma(j+m+1) \Gamma(j-m+1)}}\langle \mathbf{z}\mid \mathbf{z}\rangle_0^{i \frac{\rho}{2}-1-j}  \langle \mathbf{n}_+\mid \mathbf{z}\rangle_0^{j+m} \langle \mathbf{n}_-\mid \mathbf{z}\rangle_0^{j-m},
\end{equation}
where 
\be
\langle\mathbf{z}\mid\mathbf{z}'\rangle_{0}:=\bar{z}_{+}z'_{+}+\bar{z}_{-}z'_{-},\quad \langle\mathbf{z}\mid\mathbf{z}'\rangle:=\bar{z}_{+}z'_{+}-\bar{z}_{-}z'_{-},\quad 
\mathbf{n}_{+}=\left(\begin{array}{l}
1\\
0
\end{array}\right),\quad \mathbf{n}_{-}=\left(\begin{array}{l}
0\\
1
\end{array}\right)
\ee
and $Y_\pm$ map for $|j,m\rangle^\pm$
\begin{equation}\label{f_su11}
 F^{s}_{j m}(\mathbf{z})=\sqrt{\frac{\Gamma(s m+j)}{\pi\Gamma(2j-1) \Gamma(s m-j+1)}}\Theta(s\langle \mathbf{z}\mid\mathbf{z}\rangle)(s\langle \mathbf{z}\mid\mathbf{z}\rangle)^{i \frac{\rho}{2}-1+j}  \langle \mathbf{z}\mid\mathbf{n}_+\rangle^{-j-m} (s \langle \mathbf{z}\mid\mathbf{n}_-\rangle)^{-j+m},
\end{equation}
where $s=\pm$, and $\Theta$ is the Heaviside step function ($\Theta(x)=1$ for $x\geq 0$ and $\Theta(x)=0$ for $x<0$). The scalar products
$\langle\cdot\mid\cdot\rangle$ and $\langle\cdot\mid\cdot\rangle_0$ are SU(1,1) and SU(2) invariant hermitian inner products on $\C^2$. Derivations of \eqref{f_su2} and \eqref{f_su11} can be found in e.g. \cite{Kaminski:2017eew}.

The standard coherent states are defined in similar way as for $\Su$ as an extremal (lowest or highest) eigenvectors.
The coherent states in the standard position are thus $F_{j\pm j}^{\pm}$ and $F_{jj}$.
The general coherent state is defined by group element acting on the standard coherent state. Such states are labelled by spinors
\be
&|j,\mathbf{m}^\pm\rangle^\pm=u\act |j,\pm j\rangle^\pm ,\quad \mathbf{m}^\pm=(u^T)^{-1}\mathbf{n}_\pm,\quad u\in \mathrm{SU}(1,1)\\
&|j,\mathbf{m}^0\rangle^0=u\act |j,j\rangle^0,\quad \mathbf{m}^0=(u^T)^{-1}\mathbf{n}_+,\quad u\in \Su,
\ee
where $\mathbf{m}^\pm,\mathbf{m}^0$ satisfy
\be 
\langle \mathbf{m}^0\mid \mathbf{m}^0\rangle_0=1,\quad \langle \mathbf{m}^+\mid \mathbf{m}^{+}\rangle=1,\quad \langle \mathbf{m}^{-}\mid \mathbf{m}^{-}\rangle=-1.
\ee
The $Y$-map acting on the coherent state $|j,\mathbf{m}^\eps\rangle^\eps$ is realized by
\be
\Psi_j^{\pm,\mathbf{m}^{\pm}}(\mathbf{z})&=&\Theta\left(\pm\langle\mathbf{z}\mid\mathbf{z}\rangle\right)\sqrt{\frac{2j-1}{\pi}}(\pm\langle\mathbf{z}\mid\mathbf{z}\rangle)^{\mathrm{i}\frac{\rho}{2}-1+j}(\pm\left\langle \mathbf{z}\mid\mathbf{m}^{\pm}\right\rangle) ^{-2j},\quad\mathbf{m}^{\pm}=\left(u^{T}\right)^{-1}\mathbf{n}_{\pm},\label{coh12}\\
\Psi_j^{0,\mathbf{m}^0}(\mathbf{z})&=&\sqrt{\frac{2 j+1}{2 \pi}}\langle\mathbf{z} \mid \mathbf{z}\rangle_{0}^{\mathrm{i} \frac{\rho}{2}-1-j}\langle\mathbf{m}^0\mid \mathbf{z} \rangle_{0}^{2 j},\quad \quad\mathbf{m}^0=\left(u^{T}\right)^{-1}\mathbf{n}_{+}.\label{coh30}
\ee
More information about coherent states can be found in Appendix \ref{Coherent state}.


\subsection{Bilinear form ${\beta}$ and the scalar product}

We have realized states in the $\Slc$ unitary irrep as homogeneous functions on $\C^2$. The inner product between $\Psi_1,\Psi_2\in\ch^{(\rho,n)}$ is an integral on $\mathbb{CP}^1$
\be
\left\langle\Psi_{1}, \Psi_{2}\right\rangle=\int_{\mathbb{C} \mathbb{P}^{1}} \overline{\Psi_{1}(\mathbf{z})} \Psi_{2}(\mathbf{z}) \Omega_{\mathbf{z}},\quad \Omega_{\mathbf{z}}=\frac{\mathrm{i}}{2}\left(z_{+}dz_{-}-z_{-}dz_{+}\right)\wedge\left(\bar{z}_{+}d\bar{z}_{-}-\bar{z}_{-}d\bar{z}_{+}\right).\label{innerpsipsi}
\ee
The integral a priori defined in $\C^2$ descends to the quotient of $\mathbb{CP}^1$.

{We will use the scalar product in our definition of the vertex amplitude. There is an alternative formulation appearing in the literature.} This formulation uses the invariant bilinear form on $\ch_{(\rho,n)}$
\cite{semiclassical,gelfand5}
\be
\beta(\Phi, \Phi'):=\frac{\sqrt{\rho^2+n^2}}{2\pi}
\int_{\CP^1\times \CP^1}\Omega_{\zz}\wedge \Omega_{\zz'}|[\zz,\zz']|^{-2} [\zz,\zz']^{\frac{-\iu\rho-n} {2}}\overline{[\zz,\zz']}^{\frac{-\iu \rho+n}{2}}\ \Phi(\zz)\Phi'(\zz'),\label{betaprod}
\ee
where $[z,w] := z_0 w_1 - z_1 w_0$.

We define an antilinear map $\hat{\beta}\colon \ch_{(\rho,n)}\rightarrow \ch_{(\rho,n)}$ by
\be
\label{eq:beta-hat}
\overline{\hat{\beta}(\Phi)}(\zz')=\frac{\sqrt{\rho^2+n^2}}{2\pi}
\int_{\CP^1}\Omega_{\zz}|[\zz,\zz']|^{-2} [\zz,\zz']^{\frac{-\iu\rho-n} {2}}\overline{[\zz,\zz']}^{\frac{-\iu \rho+n}{2}}\ \Phi(\zz).
\ee
Then the invariant bilinear form relates to the inner product by
\be
\beta(\Psi,\Phi)=\langle \hat{\beta}(\Psi),\Phi\rangle\label{hat-beta}
\ee
The map is intertwining $\hat{\beta}(g \triangleright \Phi)=g \triangleright \hat{\beta}(\Phi)$ by definition \eqref{hat-beta}, and as the representation is irreducible it is proportional to an anti-unitary (in fact it is just an anti-unitary in this case \cite{gelfand5}).


\begin{Lemma}
The map $\hat{\beta}$ satisfies
\be
\hat{\beta}(\Psi_j^{\pm,\mathbf{m}^{\pm}})=C_\pm\Psi_j^{\mp,\mathbf{m}^{\mp}},\quad \hat{\beta}(\Psi_j^{0,\mathbf{m}^0})=C_0\Psi_j^{0,\mathbf{m}^0}
\ee
where
\be
C_+=\frac{\sqrt{\rho^2+n^2}}{n-i\rho},\quad C_-=(-1)^nC_+,\quad 
C_0=\frac{\sqrt{\rho^2+n^2}}{i\rho+n}
\ee
Similarly
\be
\hat{\beta}(F_{jm}^s)\propto F_{j,-m}^{-s},\quad 
\hat{\beta}(F_{jm})\propto F_{j,-m}
\ee

\end{Lemma}

The proof is provided in Appendix \ref{sec:comparison}. The Lemma allows us to introduce maps $\hat{\beta}^s\colon D_j^s\rightarrow D_j^{-s}$ and $\hat{\beta}^0\colon D_j^0\rightarrow D_j^0$ by condition
\be
\hat{\beta}Y_s=Y_s\hat{\beta}^s,\quad  \hat{\beta}Y_0=Y_0\hat{\beta}^0. \label{betas}
\ee
These maps will be useful in analysing relation between formulations (of the vertex amplitude) with the inner product \eqref{innerpsipsi} and the invariant bilinear form \eqref{betaprod}.

\subsection{4-simplex vertex amplitude}

The vertex amplitude can be realized either with use of the scalar product or invariant bilinear form $\b$. As we see above this two ways are equivalent, because we can replace contraction with bilinear form by scalar product.

Let's firstly define the vertex amplitude of a 4-simplex using coherent states \eqref{coh12} and \eqref{coh30}, then extend to more general circumstances. The vertex amplitude $A_v$ involves the coherent states acted by $g\in\Slc$: $g\act\Psi_j^{\pm,\mathbf{m}^{\pm}}(\mathbf{z})=\Psi_j^{\pm,\mathbf{m}^{\pm}}(g^T\mathbf{z})$ and $g\act\Psi_j^{0,\mathbf{m}^0}(\mathbf{z})=\Psi_j^{0,\mathbf{m}^0}(g^T\mathbf{z})$.
The vertex amplitude of a 4-simplex is expressed as
\be
A_v&=&\int \prod_{a=1}^5\rmd g_a\,\delta(g_1)\prod_{a<b}\lag \Psi_{j_{ab}}^{\eps,\mathbf{m}^{\eps}_{ab}},\,\lt(g_a^{-1}g_b\rt)\act\Psi_{j_{ab}}^{\eps',\mathbf{m}^{\eps'}_{ba}}\rag,\quad \eps,\eps'=0,\pm,\label{Avint2}\\
&=&\int \prod_{a=1}^5\rmd g_a\,\delta(g_1)\prod_{a<b}\int_{\mathbb{CP}^1}\overline{ \Psi_{j_{ab}}^{\eps,\mathbf{m}^{\eps}_{ab}}(\mathbf{z}_{ab})}\ \lt(g_a^{-1}g_b\rt)\act\Psi_{j_{ab}}^{\eps',\mathbf{m}^{\eps'}_{ba}}(\mathbf{z}_{ab})\ \O_{{\bf z}_{ab}} .\nonumber
\ee
where $a=1,\cdots,5$ labels 5 tetrahedra forming the boundary of the 4-simplex. $\delta(g_1)$ gauge fixes the redundant $\Slc$ freedom in the integrand. The vertex amplitude in canonical basis $|j,m\rangle^\eps$ ($\eps=0,\pm$) is also involved in our discussion, and is given by substituting $F_{jm},F_{jm}^\pm$ for $\Psi_{j}^{\eps,\mathbf{m}^{\eps}}$ in \eqref{Avint2}.

\subsection{Generalized vertices}

Generalizing from models on simplicial complexes, the spinfoam amplitude can be adapt to arbitrarily 4d cellular complex \cite{KKL,generalize}. The generalized amplitude is built by the vertex amplitudes whose boundaries are arbitrary valent spin-networks dual to polyhedra. These generalized amplitudes find application in spinfoam cosmology and various models where cuboid decompositions is used instead of triangulations \cite{Bahr:2017klw,Bianchi:2010zs}. There is also progress on the large-$j$ asymptotics of the generalized vertex amplitude \cite{Dona:2020yao}. Thus it is useful to take into account the generalized vertex amplitude.

Let us consider a graph $\Gamma$.
We denote by 
\begin{enumerate}
    \item $\Gamma_\fl$ the set of links in the graph,
    \item $\Gamma_\fn$ the set of nodes of the graph,
\end{enumerate}
Each link ${\fl}\in \Gamma_{\fl}$ is oriented and it starts in {$s({\fl})\in \Gamma_{\fn}$ and ends in $t({\fl})\in \Gamma_{\fn}$.} The valence of the node equals to the number of links connecting to the given node.

The important notion is connectivity of the graph. 

\begin{Definition}
The graph $\Gamma$ is disconnected if there exists a division of the set $\Gamma_{\fn}$ into two disjoint nonempty sets $V_0,V_1$
\be
V_0\cup V_1=\Gamma_{\fn},\quad V_0\cap V_1=\emptyset,\quad V_0\not=\emptyset,\quad  V_1\not=\emptyset
\ee
such that for every ${\fl}\in \Gamma_{\fl}$, both
\be
s({\fl}),t({\fl})\in V_0 \quad \text{ or }\quad  s({\fl}),t({\fl})\in V_1
\ee
\end{Definition}

We call graph connected if it is not disconnected. From the point of view of integrability the important is notion of $3$-connectivity.

\begin{Definition}\label{defconn}
The graph is $N$-connected if any graph obtained from it by removing  $k<N$ links from the set $\Gamma_{\fl}$ is still connected.
\end{Definition}

{Every $N$-connected graph is also $N'$-connected for $N'\leq N$. For examples, the 4-simplex graph is a $4$-connected graph, hence also $3$-connected and the tetrahedron graph is a $3$-connected graph.}

The generalized vertex amplitude is defined as follows. We choose
\begin{enumerate}
\item a spin $j_{\fl}$ for every link ${\fl}\in \Gamma_{\fl}$.
    \item $\eps_{{\fn}{\fl}}\in\{0,+,-\}$ for every  ${\fl}\in \Gamma_{\fl}$ and ${\fn}=t({\fl})$ or ${\fn}=s({\fl})$. 
    
\item {Any boundary state $\psi_{{\fn}{\fl}}^{\eps_{{\fn}{\fl}}}\in D_{j_{\fl}}^{\eps_{{\fn}{\fl}}}$ for each pair $(\fn,\fl)$ such that ${\fl}\in \Gamma_{\fl}$ and ${\fn}=t({\fl})$ or ${\fn}=s({\fl})$
}
\end{enumerate}
We define the amplitude by 
\be
A_\Gamma\lt(\lt\{\psi_{\fn\fl}^{\eps_{{\fn}{\fl}}}\rt\}_{\fn,\fl}\rt)=\int\limits_{\SLDC^{\Gamma_{\fn}}} \prod_{{\fn}\in \Gamma_{\fn}}\rmd g_{\fn}\,\delta(g_{{\fn}_0})\prod_{e\in \Gamma_{\fl}}\lag Y_{\eps_{t({\fl}){\fl}}} \psi_{t({\fl}){\fl}}^{\eps_{t({\fl}){\fl}}},\,\lt(g_{t({\fl})}^{-1}g_{s({\fl})}\rt)\act Y_{\eps_{s({\fl}){\fl}}}\psi_{s({\fl}){\fl}}^{\eps_{s({\fl}){\fl}}}\rag, \label{AGamma}
\ee
where ${\fn}_0\in \Gamma_{\fn}$ is a chosen node. The result is independent of this choice of ${\fn}_0$ if the integral is absolutely convergent.

\begin{Remark}\label{remarktimelike}
Both the asymptotic analysis and the Conrady-Hnybida construction suggest the following terminology which will be useful in our paper: A node of the boundary graph is the dual of a polyhedron. The faces of the polyhedra are dual to links. The polyhedron dual to $\fn$ is spacelike if $\eps_{{\fn}{\fl}}=0$ (and then we assume that all $\eps_{{\fn}{\fl}}=0$ for all $\fl$ adjacent to $\fn$), otherwise, the polyhedron is timelike (then all $\eps_{{\fn}{\fl}}=\pm$ for all $\fl$ adjacent to $\fn$). If $\eps_{{\fn}{\fl}}=+$ ($\eps_{{\fn}{\fl}}=-$) then the face has the future-pointing (past-pointing) normal. 
\end{Remark}


For the vertex amplitude with coherent states, the boundary states are parametrized by $\mathbf{m}^{\eps_{{\fn}{\fl}}}_{{\fn}{\fl}}\in \C^2$ for every ${\fl}\in \Gamma_{\fl}$ and ${\fn}=t({\fl})$ or ${\fn}=s({\fl})$, then we have $Y_{\eps_{{\fn}{\fl}}}\psi_{{\fn}{\fl}}^{\eps_{{\fn}{\fl}}}=\Psi_{j_{{\fl}}}^{\eps_{{\fn}{\fl}},\mathbf{m}^{\eps_{{\fn}{\fl}}}_{{\fn}{\fl}}}$. Here $\mathbf{m}^{\eps_{\fn \fl}}_{{\fn}{\fl}}\in \C^2$ for every ${\fn}\in \Gamma_{\fn}$ and adjacent ${\fl}\in \G_{\fl}$, representing the face normal vectors of the polyhedron dual to ${\fn}$. 
For the vertex amplitude in canonical basis, we have $\psi_{{\fn}{\fl}}^{\eps_{{\fn}{\fl}}}=|j_{l},m_{\fn\fl}\rangle^{\eps_{\fn\fl}}$ and $Y_{\eps_{{\fn}{\fl}}}\psi_{{\fn}{\fl}}^{\eps_{{\fn}{\fl}}}=F^{\eps_{\fn\fl}}_{j_{l}m_{\fn\fl}}$ ($F^0_{jm}\equiv F_{jm}$).

The proof of the absolute convergence of the integral applies to either coherent or canonical-basis boundary states. As these states are preserved by $\hat{\beta}$-map, our proof gives also absolutely convergent in the case of the vertex amplitude with such states written in terms of bilinear map
\be
A_\Gamma^\beta\lt(\lt\{\psi_{\fn\fl}^{\eps_{{\fn}{\fl}}}\rt\}_{\fn,\fl}\rt)=\int\limits_{\SLDC^{\Gamma_{\fn}}} \prod_{{\fn}\in \Gamma_{\fn}}\rmd g_{\fn}\,\delta(g_{{\fn}_0})\prod_{e\in \Gamma_{\fl}}\beta\left( Y_{\eps_{t({\fl}){\fl}}} \psi_{t({\fl}){\fl}}^{\eps_{t({\fl}){\fl}}},\,\lt(g_{t({\fl})}^{-1}g_{s({\fl})}\rt)\act Y_{\eps_{s({\fl}){\fl}}}\psi_{s({\fl}){\fl}}^{\eps_{s({\fl}){\fl}}}\right)\label{AGamma-beta}
\ee
{In the case of \eqref{AGamma-beta}
we assume that the integral in the bilinear form $\beta$ is evaluated first, because this integral is finite but not absolutely convergent when taking into account the integral in $\beta$. }

Because of the relation \eqref{hat-beta} between $\langle\cdot,\cdot\rangle$ and $\b(\cdot,\cdot)$ and the properties \eqref{betas} of $\hat{\b}^0,\hat{\b}^\pm$, we have  
\be
A_\Gamma^\beta\lt(\lt\{\psi_{t(\fl)\fl}^{\eps_{t(\fl)\fl}},\psi_{s(\fl)\fl}^{\eps_{s(\fl)\fl}}\rt\}\rt)=A_\Gamma\lt(\lt\{\hat{\b}^{\eps_{t(\fl)\fl}}\lt(\psi_{t(\fl)\fl}^{\eps_{t(\fl)\fl}}\rt),\psi_{s(\fl)\fl}^{\eps_{s(\fl)\fl}}\rt\}\rt).\label{AAbeta}
\ee
One can define spinfoam amplitude equivalently using either $A_\G$ or $A_\G^\b$, although they are parametrized differently. 


\section{Absolute convergence}\label{sec3}

We are going to show that the intergal \eqref{AGamma} is absolutely convergent.
We identify the upper unit hyperboloid $\ch$ with the quotient of the group $\SLDC$
\be
\ch=\SLDC/\Su,\quad [g]\in\ch,\quad [gu]=[g]\text{ for } u\in \Su.
\ee
From any function $f$ on $\SLDC$, we produce a function on $\ch$ by integration
\be
\tilde{f}([g])=\int_{\Su} d\mu(u)\ f(gu),
\ee
where $d\mu$ denotes the Haar measure on $\Su$. Element $(g^{-1}g')^T\in\Slc$ can be expressed as the product of two $\Su$ matrices and a diagonal matrix
\be
(g^{-1}g')^{T}=R'\left(\begin{array}{cc}
\lambda & 0\\
0 & \lambda^{-1}
\end{array}\right)R,\quad\lambda=e^{d([g],[g'])/2},
\ee
where $d([g],[g'])$ is a distance on a unit hyperboloid between points $[g]$ and $[g']$ and $R,R'\in \Su$. Let us notice
\be
d([g^{-1}g'],0)=d([g'],[g]).
\ee
The following is an abstract version of the theorem proven for $4$-simplex and several other graphs in \cite{Baez:2001fh} and extended to all $3$-connected graphs in \cite{Kaminski:2010qb}.

\begin{Theorem}\label{thm-1}
Let $\Gamma$ be a $3$-connected graph and
let for any link ${\fl}\in \Gamma_{\fl}$ in the graph $K_{\fl}(x,y)$ be measureable function on $\ch\times \ch$ with property
\be
\forall\ {\fl\in \Gamma_{\fl}},\ \exists\ {C>0}, \quad |K_{\fl}(x,y)|\leq C \frac{d(x,y)}{\sinh d(x,y)} 
\ee
then the integral
\be
\int_{\ch^{\Gamma_{\fn}}}\prod_{{\fn}\in \Gamma_{\fn},\ \fn\not=\fn_0} d\mu(x_{\fn})\ \prod_{e\in \Gamma_{\fl}}K_{\fl}(x_{t({\fl})},x_{s({\fl})})
\ee
is absolutely convergent. Here $d\mu$ denotes $\SLDC$ invariant measure on $\ch$, and $\fn_0$ is a chosen node.
\end{Theorem}

\begin{proof}
In this case 
\be
|K_{\fl}(x,y)|\leq C {\ck_0(x,y),\quad \ck_0(x,y)=\frac{d(x,y)}{\sinh d(x,y)}}
\ee
where $\ck_0$ is a kernel for Barrett-Crane model with $\rho=0$ \cite{BC,Baez:2001fh}. The absolute convergence follows from the absolute convergence for Barrett-Crane model which is covered by theorem of \cite{Kaminski:2010qb}.
\end{proof}

Based on that we prove the following:

\begin{Theorem}\label{thm-2}
Let $\Gamma$ be a $3$-connected graph and for all links ${\fl}\in \Gamma_{\fl}$ in the graph, 
let $K_{\fl}(g,h)$ of ${\fl}\in \Gamma_{\fl}$ be measureable functions on $\SLDC\times \SLDC$ with property
\be
\forall\ {\fl\in \Gamma_{\fl}},\ \exists\ {C>0},\quad \lt|K_{\fl}(g,h)\rt|\leq C \frac{d([g],[h])}{\sinh d([g],[h])}
\ee
then the integral
\be
\int_{\SLDC^{\Gamma_{\fn}}}\prod_{{\fn}\in \Gamma_{\fn},\ \fn\not=\fn_0} dg_{\fn}\ \prod_{e\in \Gamma_{\fl}}K_{\fl}(g_{t({\fl},\ \fn\not=\fn_0)},g_{s({\fl})})
\ee
is absolutely convergent. Here $\fn_0$ is a chosen node.
\end{Theorem}

\begin{proof}
We introduce
\be
\tilde{K}_{\fl}([g],[g'])=\sup_{u,u'\in \Su} |K_{\fl}(gu,g'u)|
\ee
{We introduce a notation for the set of nodes different than $\fn_0$
\be
\Gamma_{\fn}^\ast:=\Gamma_{\fn}\setminus \{\fn_0\}
\ee
}
We can estimate
\be\label{eq:K-tildeK}
\int_{\SLDC^{\Gamma_{\fn}^\ast}}\prod_{{\fn}\in \Gamma_{\fn}^\ast} dg_{\fn}\ \prod_{e\in \Gamma_{\fl}}\lt|K_{\fl}(g_{t({\fl})},g_{s({\fl})})\rt|\leq
\int_{\SLDC^{\Gamma_{\fn}^\ast}}\prod_{\fn\in \Gamma_{\fn}^\ast} dg_{\fn}\ \prod_{e\in \Gamma_{\fl}}\tilde{K}_{\fl}([g_{t({\fl})}],[g_{s({\fl})}])
\ee
The integration over $\SLDC$ can be written as integration over $\Su$ and $\ch$
\be
\int_{\SLDC}dg=\int_{\Su}d\mu(u)\int_{\ch}d\mu(x)
\ee
thus the right-hand side of equation \eqref{eq:K-tildeK} can be written as
\be
&\int_{\Su^{\Gamma_{\fn}^\ast}}\prod_{{\fn}\in \Gamma_{\fn}^\ast} d\mu(u_{\fn})\ \int_{\ch^{\Gamma_{\fn}^\ast}}\prod_{{\fn}\in \Gamma_{\fn}^\ast}d\mu(x_{\fn})\ \prod_{e\in \Gamma_{\fl}}\tilde{K}_{\fl}(x_{t({\fl})},x_{s({\fl})})=\nonumber\\
&=|\Su|^{|\Gamma_{\fn}|-1}\int_{\ch^{\Gamma_{\fn}^\ast}}\prod_{{\fn}\in \Gamma_{\fn}^\ast} d\mu(x_{\fn})\ \prod_{e\in \Gamma_{\fl}}\tilde{K}_{\fl}(x_{t({\fl})},x_{s({\fl})})
\ee
where $|\Su|$ is the total volume of the $\Su$ group. As $\tilde{K}_{\fl}$ satisfies the bound
\be
|\tilde{K}_{\fl}(x,x')|\leq C \frac{d(x,x')}{\sinh d(x,x')}
\ee
we can apply
the previous theorem \ref{thm-1}.
\end{proof}

\subsection{Bounded functions}

Our goal is to show that the function
\be
K_{\fl}(g_{t({\fl})},g_{s({\fl})})=\lag Y_{\eps_{t({\fl}){\fl}}} \psi_{t({\fl}){\fl}}^{\eps_{t({\fl}){\fl}}},\,\lt(g_{t({\fl})}^{-1}g_{s({\fl})}\rt)\act Y_{\eps_{s({\fl}){\fl}}}\psi_{s({\fl}){\fl}}^{\eps_{s({\fl}){\fl}}}\rag
\ee
has the property assumed in the theorem \ref{thm-2}. Let us denote
\be
|\mathbf{z}|_0=\sqrt{\langle \zz\mid\zz\rangle_0}
\ee
In order to state our approach in more general setting we introduce a class of functions:

\begin{Definition}\label{bfunc}
Let us consider a measurable function of two spinor $f(\mathbf{z})$. We say that $f\in L_\infty$ if there exists a constant $C>0$ such that for almost all $\mathbf{z}$ (with respect to the Lebesgue measure)
\be\label{eq:Linfty}
|f(\mathbf{z})|\leq C |\mathbf{z}|_0^{-2} 
\ee
We define the norm
\be
|f|_\infty=\operatorname{ess}\operatorname{sup}_{\C^2}|f(\mathbf{z})||\mathbf{z}|_0^2
\ee
where the essential supremum is defined with respect to the Lebegue measure. \end{Definition}

Let us remind that the essential supremum of a real function $h_\R$ is defined as
\be
\operatorname{ess}\operatorname{sup}_{\C^2} h_\R(\zz)=\inf\left\{C\in \R\colon h(\zz)\leq C\text{ for almost all } \zz\in \C^2\right\}
\ee
In particular if a function $f$ is  continuous except at origin then
\be
|f|_\infty=\operatorname{sup}_{0\not=\zz\in \C^2}|f(\mathbf{z})||\mathbf{z}|_0^2
\ee

Let us notice that if $f(\mathbf{z})$ is such that for almost all $\mathbf{z}$
\be
\forall\ {r,\phi},\quad 
f(re^{i\phi}\mathbf{z})=r^{i\rho-2}e^{i2j\phi}f(\mathbf{z}),
\ee
then the property \eqref{eq:Linfty} is equivalent to
\be
|f(\mathbf{z})|\leq C,\text{ for almost all } |\mathbf{z}|_0=1
\ee
where the measure on the sphere $|\mathbf{z}|_0=1$ is again the standard Lebegue measure. The norm can be computed by
\be
|f|_\infty=\operatorname{ess}\operatorname{sup}_{|\mathbf{z}|_0=1}|f(\mathbf{z})|
\ee
The space of functions which belong to both $L_\infty$ and $\ch_{(\rho,n)}$ will be denoted by $L_\infty(\ch_{(\rho,n)})$.

The importance of this definition is based on some facts:

\begin{Lemma}\label{lm:gact}
Given $g\in\Slc$, if $\Psi\in L_\infty$ then $g\act \Psi\in L_\infty$. Moreover, if $g\in \Su$ then $|g\act \Psi|_\infty=|\Psi|_\infty$
\end{Lemma}

\begin{proof}
Fixing $g\in\Slc$, there exists a constant $E>0$ such that
\be
|g^T\mathbf{z}|_0\geq E
\ee
for $|\mathbf{z}|_0=1$ as the sphere is compact. Thus
\be
|g^T\mathbf{z}|_0\geq E|\mathbf{z}|_0
\ee
in general. 
Then
\be
|\Psi(g^T\mathbf{z})|\leq C|g^T\mathbf{z}|_0^{-2}\leq CE^{-2}|\mathbf{z}|_0^{-2}
\ee
and $g\act \Psi\in L_\infty$. Moreover as
\be
|u^T\mathbf{z}|_0=|\mathbf{z}|_0
\ee
for $u\in \Su$ we also have $|u\act \Psi|_\infty=|\Psi|_\infty$.
\end{proof}

\begin{Lemma}\label{lm:Fjm}
For any $m\geq j$, $F_{j,\pm m}^\pm\in L_\infty$ as well as for any $-j\leq m\leq j$, $F_{jm}\in L_\infty$.
\end{Lemma}

\begin{proof}
We consider first $Y^\pm$ map for $SU(1,1))$.  Let us choose $s\in \{-,+\}$. The following inequalities hold
\begin{align}
|z_+|^2-|z_-|^2\leq |z_+|^2&\qquad\Longrightarrow&
\langle \mathbf{z}\mid\mathbf{z}\rangle \leq|\langle \mathbf{z}\mid  \mathbf{n}_+\rangle |^2\\
-|z_+|^2+|z_-|^2\leq |z_-|^2&\qquad\Longrightarrow& 
(-\langle \mathbf{z}\mid\mathbf{z}\rangle) \leq |\langle \mathbf{z}\mid  \mathbf{n}_-\rangle|^2
\end{align}
This shows that
\be
(s\langle \mathbf{z}\mid\mathbf{z}\rangle) \leq |\langle \mathbf{z}\mid  \mathbf{n}_s\rangle|^2
\ee
Suppose that $s\langle \mathbf{z}\mid\mathbf{z}\rangle \geq 0$ then it is also true that
\be
2|\langle \mathbf{z}\mid \mathbf{n}_s\rangle |^{2}=2|z_s|^2= |\mathbf{z}|_0^2+s\langle\mathbf{z}\mid \mathbf{z}\rangle\geq |\mathbf{z}|_0^2
\ee
thus
\be
|\langle \mathbf{z}\mid \mathbf{n}_s\rangle |^{-2}\leq 2 |\mathbf{z}|_0^{-2}
\ee
Moreover, as $|\langle \mathbf{z}\mid\mathbf{n}_+\rangle|^2-|\langle \mathbf{z}\mid\mathbf{n}_-\rangle|^2=\langle \mathbf{z}\mid\mathbf{z}\rangle$
we have
\be
|\langle \mathbf{z}\mid\mathbf{n}_s\rangle|\geq |\langle \mathbf{z}\mid\mathbf{n}_{-s}\rangle|\text{ if }s\langle  \mathbf{z}\mid\mathbf{z}\rangle\geq 0
\ee
Consider now separately $|F_{jm}^s|$, $s=\pm$ (skipping the constant) on its support, which due to Heaviside's step function is given by $s\langle \mathbf{z}\mid\mathbf{z}\rangle \geq 0$:
\begin{enumerate}
    \item $s=+$ and $m\geq j$ then for $\langle \mathbf{z}\mid\mathbf{z}\rangle\geq 0$
    \be
    \langle \mathbf{z}\mid\mathbf{z}\rangle^{j-1}\leq |\langle \mathbf{z}\mid\mathbf{n}_+\rangle|^{2j-2},\quad
    |\langle \mathbf{z}\mid\mathbf{n}_-\rangle|^{-j+m}\leq |\langle \mathbf{z}\mid\mathbf{n}_+\rangle|^{-j+m}
    \ee
    because $m\geq j$ and $j\geq 1$
    \be
\langle \mathbf{z}\mid\mathbf{z}\rangle^{j-1}  |\langle \mathbf{z}\mid\mathbf{n}_+\rangle|^{-j-m} | \langle \mathbf{z}\mid\mathbf{n}_-\rangle|^{-j+m}\leq |\langle \mathbf{z}\mid\mathbf{n}_+\rangle|^{-2}\leq 2|\mathbf{z}|_0^{-2}
\ee
    \item $s=-$ and $m\leq -j$ then for $\langle \mathbf{z}\mid\mathbf{z}\rangle\leq 0$
    \be
    (-\langle \mathbf{z}\mid\mathbf{z}\rangle)^{j-1}\leq |\langle \mathbf{z}\mid\mathbf{n}_-\rangle|^{2j-2},\quad
    |\langle \mathbf{z}\mid\mathbf{n}_+\rangle|^{-j-m}\leq |\langle \mathbf{z}\mid\mathbf{n}_-\rangle|^{-j-m}
    \ee
    because $m\leq -j$ ($-j-m\geq 0$) and $j\geq 1$
    \be
(-\langle \mathbf{z}\mid\mathbf{z}\rangle)^{j-1}  |\langle \mathbf{z}\mid\mathbf{n}_+\rangle|^{-j-m} | \langle \mathbf{z}\mid\mathbf{n}_-\rangle|^{-j+m}\leq |\langle \mathbf{z}\mid\mathbf{n}_-\rangle|^{-2}\leq 2|\mathbf{z}|_0^{-2}
\ee
\end{enumerate}
This shows that $F_{jm}^s\in L^\infty$.

States $F_{jm}$ are just smooth function on the sphere $|\zz|_0=1$ so they are bounded. However, as before we can compute $|F_{jm}|$ skipping the constant (using $|\langle \mathbf{m}\mid \mathbf{z}  \rangle_{0}|^2\leq |\langle\mathbf{z} \mid \mathbf{z}\rangle_{0}|$)
\be
\langle \mathbf{z}\mid \mathbf{z}\rangle_0^{-1-j}  |\langle \mathbf{n}_+\mid \mathbf{z}\rangle_0|^{j+m} |\langle \mathbf{n}_-\mid \mathbf{z}\rangle_0|^{j-m}\leq |\langle\mathbf{z} \mid \mathbf{z}\rangle_{0}|^{-1}=|\mathbf{z}|_0^{-2}
\ee
so $F_{jm}\in L_\infty$. 
\end{proof}

\begin{Lemma}
The coherent states $\Psi_j^{\pm,\mathbf{m}^{\pm}}$ and $\Psi_j^{0,\mathbf{m}^0}$ belong to $L_\infty$.
\end{Lemma}

\begin{proof}From Lemma \ref{lm:Fjm} functions $F_{jj}$ and $F_{j\pm j}^\pm$ belong to $L_\infty$. The coherent states are given by
\begin{equation}
    \Psi_j^{s,\mathbf{m}^s}=g\act F_{j,sj}^s\in L_\infty,\quad 
    \Psi_j^{0,\mathbf{m}^0}=g\act F_{j,j}\in L_\infty
\end{equation}
due to Lemma \ref{lm:gact},
where $g\in SU(1,1)$ is given by
\begin{equation}
    g\mathbf{n}_s=\mathbf{m}^s,\quad g\mathbf{n}_+=\mathbf{m}^0
\end{equation}
depending on which case we are considering.
\end{proof}

The above 2 lemmas show that the canonical basis $F_{j,m},F^\pm_{j,m}$ and coherent state $\Psi_j^{0,\mathbf{m}^0},\Psi_j^{\pm,\mathbf{m}^\pm}$ share the same property, which is the key for proving the absolute convergence of the integral.

\begin{Remark}
Let us notice that even if $\Psi\in L_\infty$, it might happen that $\hat{\beta}(\Psi)\notin L_\infty$ for general $\Psi$, thus our proof of convergence does not translate verbatim to the proof of the vertex amplitude in terms of bilinear scalar product unless we assume additionally that $\hat{\beta}(\Psi)\in L_\infty$, for boundary state. This assumption is however satisfied for both coherent states and the canonical basis as proven in Appendix \ref{sec:comparison}. This is sufficient for our applications.
\end{Remark}

\subsection{Integral on \texorpdfstring{$\mathbb{CP}^1$}{CP1} and the bound}

We can introduce a parametrization of a cross section $\CP^1$
\be
\zz=\left(\begin{array}{c}
-e^{i\phi/2}\sin \theta/2\\e^{-i\phi/2}\cos \theta/2
\end{array}\right),\quad \theta\in [0,\pi],\quad \phi\in [0,2\pi)
\ee
This parametrization is singular only at $\theta=0,\pi$ (measure zero sets). 
\be
z_+ d z_--z_-dz_+=\frac{1}{2}(d\theta+i\sin\theta d\phi)
\ee
The measure is given by
\be
\Omega_z=\frac{1}{4}\sin\theta\, d\theta\wedge d\phi
\ee
The scalars product as the integrals over $\mathbb{CP}^1$ can be expressed by
\be\label{eq:CP1-theta}
\langle \Phi,\Psi\rangle=\frac{1}{4}\int_0^{2\pi}d\phi\int_{0}^\pi d\theta\, \sin\theta\,  \overline{\Phi(\mathbf{z})}\,\Psi(\mathbf{z}),\qquad \zz=\left(\begin{array}{c}
-e^{i\phi/2}\sin \theta/2\\e^{-i\phi/2}\cos \theta/2
\end{array}\right)
\ee
Now the main theorem:
\begin{Theorem}\label{thm-3}
Let $\Phi,\Psi\in L_\infty(\ch_{(\rho,n)})$ then
\be
|\langle \Phi, g\act \Psi\rangle |\leq 2\pi|\Phi|_\infty |\Psi|_\infty\ \frac{d([g],0)}{\sinh d([g],0)},
\ee
\end{Theorem}

\begin{proof}
First let us notice that if we write
\be
g^{T}=R'BR,\quad B=\left(\begin{array}{cc}
\lambda & 0\\
0 & \lambda^{-1}
\end{array}\right),\ \lambda=e^{d([g],0)/2}\geq1
\ee
where $R,R'\in \Su$ then 
\be
\langle \Phi, g\act \Psi\rangle=\langle R^{-1}\act\Phi, B\act (R'\act\Psi)\rangle
\ee
Introducing $\Phi'=R\act\Phi$ ($|\Phi'|_\infty=|\Phi|_\infty$) and
$\Psi'=R'\act\Psi$ ($|\Psi'|_\infty=|\Psi|_\infty$) we should show the inequality only for $g=B$. 
We estimate
\be
|\Phi(\mathbf{z})|\leq |\Phi|_\infty|\mathbf{z}|^{-2}_0,\quad |\Psi(B^T\mathbf{z})|\leq |\Psi|_\infty|B^T\mathbf{z}|^{-2}_0, 
\ee
for almost all $\mathbf{z}$.
Moreover, we compute
\be
|\zz|_0^2=1,\quad |B^T\mathbf{z}|^{2}_0=\lambda^2\sin^2(\theta/2)+\lambda^{-2}\cos^2(\theta/2),\quad \zz=\left(\begin{array}{c}
-e^{i\phi/2}\sin \theta/2\\e^{-i\phi/2}\cos \theta/2
\end{array}\right).
\ee
This allows to estimate the integral \eqref{eq:CP1-theta} as follows
\be
|\langle \Phi, B\act \Psi\rangle|\leq |\Phi|_\infty |\Psi|_\infty \frac{1}{4}\int_0^{2\pi}d\phi\int_{0}^\pi d\theta\ \sin\theta\ (\lambda^2\sin^2(\theta/2)+\lambda^{-2}\cos^2(\theta/2))^{-1}.
\ee
The following equality holds for {$\l\geq1$}
\be
\int_0^{2\pi}d\phi\int_{0}^\pi d\theta\ \sin\theta\ (\lambda^2\sin^2(\theta/2)+\lambda^{-2}\cos^2(\theta/2))^{-1}
=16\pi \frac{\ln \lambda}{\lambda^2-\lambda^{-2}}.
\ee
Additionally $\ln \lambda=d([g],0])/2$ and 
\be
\lambda^2-\lambda^{-2}=2\sinh d([g],0]),
\ee
thus
\be
|\langle \Phi, B\act \Psi\rangle|\leq 2\pi |\Phi|_\infty |\Psi|_\infty \frac{d([g],0)}{\sinh d([g],0)}.
\ee
We showed the inequality in the special case $g=B$.
\end{proof}

Finally we have the absolute convergence of the vertex amplitude

\begin{Theorem}\label{thm:main}
Let $\Gamma$ be a $3$-connected spin-network graph with generalized vertex data (defined above \eqref{AGamma}). Suppose that the boundary states $\psi_{\fn\fl}^{\eps_{\fn\fl}}$ are either the canonical-basis or coherent states of $\Su$ or $\mathrm{SU}(1,1)$. Then the integral \eqref{AGamma} for vertex amplitude $A_\Gamma(\{\psi_{\fn\fl}^{\eps_{\fn\fl}})\})$ is absolutely convergent.
\end{Theorem}

\begin{proof}
Using theorem \ref{thm-3} we can estimate
\be
\lt|K_{{\fl}}(g_{t({\fl})},g_{s({\fl})})\rt|=\left|\lag \Phi,\,\lt(g_{t({\fl})}^{-1}g_{s({\fl})}\rt)\act\Psi\rag\right|
\ee
by the quantity
\be
2\pi\lt|\Phi\rt|_\infty \lt|\Psi\rt|_{\infty} \frac{d([g_{t({\fl})}^{-1}g_{s({\fl})}],0)}{\sinh d([g_{t({\fl})}^{-1}g_{s({\fl})}],0)}
\ee
where $\Psi,\Phi$ are either the canonical basis $F_{j,m},F^\pm_{j,m}$ or coherent states $\Psi_j^{0,\mathbf{m}^0},\Psi_j^{\pm,\mathbf{m}^\pm}$. Thanks to $d([g^{-1}h],0)=d([g],[h])$, we obtain
\be
\lt|K_{{\fl}}(g_{t({\fl})},g_{s({\fl})})\rt|\leq
2\pi\lt|\Phi\rt|_\infty \lt|\Psi\rt|_{\infty} 
\frac{d([g_{t({\fl})}],[g_{s({\fl})}])}{\sinh d([g_{t({\fl})}],[g_{s({\fl})}])}
\ee
This finishes the proof by theorem \ref{thm-2}.
\end{proof}

\begin{Corollary}
The 4-simplex vertex amplitude $A_v$ in Eq.\eqref{Avint2} is absolute convergent.
\end{Corollary}
 
\begin{proof}
The convergence of the 4-simplex vertex amplitude follows from Theorem \ref{thm:main}, because the $4$-simplex graph is $3$-connected (recall Definition \ref{defconn}).   
\end{proof}

In the applications it is often more useful to use vertex amplitude with bilinear form. It differs by application of $\hat{\beta}^\eps$ map to some of the states (recall \eqref{AAbeta}). However the family of states considered in Theorem \ref{thm:main} is preserved by this map thus we obtain a simple Corollary:

\begin{Corollary}\label{cor:main}
Under the assumptions of Theorem \ref{thm:main}, the integral \eqref{AGamma-beta} for  vertex amplitude $A_\Gamma^\beta(\{\psi_{\fn\fl}^{\eps_{\fn\fl}})\}$ is absolutely convergent.
\end{Corollary}

\section{Gluing vertex amplitudes}\label{Gluing vertex amplitudes}

In the case of the EPRL model there are two equivalent ways to glue vertex amplitudes and define spinfoam amplitudes on complexes. One can either sum and integrate over all intermediate states with a proper gauge fixing {(e.g. the integral formulation in \cite{Rovelli:2010vv,HZ,hanPI})}, or one can restrict to the space of $\Su$ intertwiners. This two method are obviously equivalent as the space of intertwiners is a subspace to which we nevertheless project. The situation is quite different in the case of SU(1,1) intertwiners, as they are no longer elements of the original Hilbert space. It is not obvious that these 2 choices of formulations can be equivalent in general for SU(1,1). We will provide in this section the definition of the amplitude on complex in terms of intertwiner space.

Although we prove the finiteness of the vertex amplitude, the finiteness of spinfoam amplitudes on complexes with fixed $j$'s are nontrivial, when vertex amplitudes are glued with timelike tetrahedra/polyhedra, because the SU(1,1) intertwiner space is generally infinite-dimensional. The situation is much more complicated than the EPRL model, where the finiteness of vertex amplitudes implies the finiteness of amplitudes with fixed $j$'s on any complex, due to the finite-dimensional intertwiner space of SU(2).

We are going to find that, in certain circumstances, the SU(1,1) intertwiner space becomes finite-dimensional.
When vertex amplitudes in this case are glued with a spacelike tetrahedron or polyhedron which corresponding to a finite-dimensional SU(2) intertwiner space, the sum of intertwiners is finite thus does not introduce any divergence, similar to the EPRL model. In such a case, the finiteness of the space of intertwiners implies the finiteness of the extended spinfoam amplitude with fixed $j$'s.

\subsection{Intertwiner space}

Given two SU(1,1) representations $D^+_{j_1}$ and $D^+_{j_2}$, we have the decomposition of their tensor product
\begin{equation}
    D_{j_1}^+\otimes D_{j_2}^+=\oplus_{j\geq j_1+j_2} D_j^+.
\end{equation}
We introduce {orthogonal projections}
$I^{j_1j_2}_j\colon D_{j_1}^+\otimes D_{j_2}^+\rightarrow D_j^+$ such that 
\begin{equation}
    I^{j_1j_2}_j|j_1,m_1\rangle_+\otimes |j_2,m_2\rangle_+=\alpha_{m_1+m_2} |j_1+j_2,m_1+m_2\rangle_+, 
\end{equation}
{for some coefficients $\alpha_{m_1+m_2}\in \C$.} We have the embedding
\begin{equation}
    {I^{j_1j_2}_j}^\dagger |j,m\rangle_+=\sum_{m_1+m_2=m} \alpha_{m_1,m_2}
    |j_1,m_1\rangle_+\otimes |j_2,m_2\rangle_+,
\end{equation}
for some $\alpha_{m_1,m_2}\in \C$. It is important that this sum is \emph{finite} as $m_1\geq j_1$ and $m_2\geq j_2$.

We introduce subspaces $S_j^\pm\subset D_j^\pm$ which will be useful in our analysis of the space of invariants
\begin{align}
    & S_j^\pm=\text{alg. span}\{{ u\act }|j,m \rangle_\pm,\ u\in SU(1,1),\ \pm m\geq j\},
\end{align}
which contain all \emph{finite} linear combinations of $u\act |j,m \rangle_\pm$. We know by Lemma \ref{lm:Fjm} that
\begin{equation}
    F_{j,m}^\pm=Y_\pm(|j,m\rangle_\pm)
\end{equation}
belongs to $L^\infty$ (also after acting upon by any group element). This means that elements of $S_j^\pm$ as a boundary states yield finite vertex amplitude. {We also see that 
\begin{equation}
    I^{j_1j_2}_j(u\act |j,m\rangle_+)=\sum_{m_1+m_2=m} \alpha_{m_1,m_2}
    u\act |j_1,m_1\rangle_+\otimes u\act |j_2,m_2\rangle_+,
\end{equation}
thus
\begin{equation}\label{eq:ISS-S}
    {I^{j_1j_2}_j}^\dagger(S_{j}^+)\subset S_{j_1}^+\hat{\otimes} S_{j_2}^+.
\end{equation}
where $\hat{\otimes}$ is an algebraic tensor product\footnote{Algebraic tensor product is the finite linear span of simple tensors without completion. We now see that every elements in
\begin{equation}
    Y\left( 
    \bigotimes_i^\wedge S_{j_i^+}^+\hat{\otimes} \bigotimes_i^\wedge S_{j_i^-}^-\right),\quad Y=\left( 
    \bigotimes_i Y_+{\otimes} \bigotimes_i Y_-\right)
\end{equation}
belongs to $L^\infty$ (as products and finite sums of bounded functions are bounded).}.}


The decomposition of multiple tensor products into direct product of irreducible representation can be done by application of pair decomposition. We introduce
\begin{equation}\label{eq:D-decomp}
    D_{j_1}^+\otimes D_{j_2}^+\otimes \ldots D_{j_n}^+=\bigoplus_{k=1}^{\infty} D_{j^+(k)}^+
\end{equation}
where we choose a function $j^+(k)$ to label representations in the decomposition. It is possible to have $j^{+}(k)=j^{+}(k')$ for $k\neq k'$, which indicates that the multiplicity at $j^+(k)$ is not 1. We define 
\begin{equation}
    I^+_k\colon D_{j_1}^+\otimes D_{j_2}^+\otimes \ldots D_{j_n}^+\rightarrow D_{j^+(k)}^+
\end{equation}
{
Let us denote
\begin{align}
    &S:=\bigotimes_i^\wedge S_{j_i^+}^+\hat{\otimes} \bigotimes_i^\wedge S_{j_i^-}^-=S_{j_1^+}^+\hat{\otimes} S_{j_2^+}^+\hat{\otimes}\ldots  S_{j_n^+}^+\hat{\otimes}S_{j_1^-}^-\hat{\otimes} S_{j_2^-}^-\hat{\otimes}\ldots  S_{j_m^-}^-
\end{align}
Inductive application of \eqref{eq:ISS-S} shows the following inclusion
\be
 {I^+_k}^\dagger(S_{j^+(k)}^+)\subset S.
\ee
Every elements in $Y(S)$ belongs to $L^\infty$.} Similar construction for $D_j^-$ gives $I^-_k$ with the same properties.

Representations $D_j^+$ and $D_j^-$ can be related by the anti-unitary intertwiner $J$ maps $|j,j\rangle^+$ to $|j,-j\rangle^-$ and vice versa. $J$ satisfies the intertwining property $g\act J(\Psi^\pm)=J(g\act\Psi^\pm)$, $\forall\, \Psi^\pm\in D^\pm_j$ ($J$ is proportional to $\hat{\b}^\pm$), which combines the anti-unitarity and leads to
\begin{equation}
    \lansu \Psi,g\act \Psi'\ransu=\overline{\lansu J(\Psi),g\act J(\Psi')\ransu},\quad \forall\ \Psi,\Psi'\in D^\pm_j .
\end{equation}
where $\lansu\cdot,\cdot\ransu$ is a hermitian scalar product in $D_j^\pm$. The Schur othogonality relations of Wigner D-matrices ($\mathrm{D}_{m_1,m_2}^{j}(g) = \langle j \, m_1| g | j , m_2 \rangle_{\rm SU(1,1)}$ for either $D_j^+$ or $D_j^-$ representations)
\be
\int_{\rm SU(1,1)} dg\, \overline{{\rm D}_{m,n}^{j}(g)}\, {\rm D}_{m',n'}^{j'}(g)=\frac{1}{d_{j}}\delta^{j,j'}\delta_{m,m'}\delta_{n,n'},\quad d_j=2j-1 \label{schur_relation}
\ee
implies that for any $\Psi_1,\Psi_1'\in D_{j}^+$, $\Psi_2,\Psi_2'\in D_{j'}^-$, 
\begin{equation}\label{eq:orthogonality}
    \int_{\rm SU(1,1)} dg\ \lansu \Psi_1,g\act \Psi_1'\ransu
    \lansu \Psi_2,g\act \Psi_2'\ransu=
    \frac{1}{d_j}\delta^{jj'}\lansu \Psi_1,J(\Psi_2)\ransu \lansu J(\Psi_1'),\Psi_2'\ransu,
\end{equation}
where $dg$ is the Haar measure.


{
Our approach to the space of invariants is through group averaging technique. The space of invariants is not a subspace of the tensor product of representations because SU(1,1) is non-compact. We will regard invariants as antilinear functionals on the subspace $S$. For any element of $\Psi\in S$ we define such a functional by
\begin{equation}
\forall\ \Phi\in S=\bigotimes_i^\wedge S_{j_i^+}^+\hat{\otimes} \bigotimes_i^\wedge S_{j_i^-}^-,\quad    [\Psi](\Phi)=\int_{\rm SU(1,1)} dg\ \lansu \Phi,g\act \Psi\ransu.\label{functionalPsi}
\end{equation}
For completeness we also introduce
\begin{equation}
[\Psi]^\dagger(\Phi):=\overline{[\Psi](\Phi)}=[\Phi](\Psi).
\end{equation}
The last equality follows from 
\begin{equation}
\int_{\rm SU(1,1)} dg\ \lansu \Phi,g\act \Psi\ransu=\int_{\rm SU(1,1)} dg\ \lansu g^{-1}\act\Phi, \Psi\ransu=
\int_{\rm SU(1,1)} dg\ \lansu g\act\Phi, \Psi\ransu.
\end{equation}
We will show in Lemma \ref{lm:well-defined} that these are indeed a well-defined functionals. The group averaging allows us to introduce natural hermitian product for two such functionals
\begin{equation}
    \laninv [\Phi],[\Psi]\raninv:=[\Psi]( \Phi)\label{invinnerprod}
\end{equation}
Let us notice that if we have
\be
\forall\ {\Phi \in S},\quad \laninv [\Phi],[\Psi]\raninv=0
\ee
then $[\Psi]=0$ as a functional, thus the scalar product is well-defined and non-degenerate. The nontrivial fact which will be proven in Lemma \ref{lm:inv} is that it is also positively defined. After saying that we introduce a definition:

\begin{Definition}
The intertwiner space 
\begin{equation}
    \operatorname{Inv}\left(\bigotimes_i
    D_{j_i^+}^+\otimes \bigotimes_i D_{j_i^-}^-\right)=\{[\Psi]\colon \Psi\in S\}^{cl}
\end{equation}
where the closure $\{\cdots\}^{cl}$ is with respect to the Hilbert space structure  $\laninv \cdot,\cdot\raninv$ defined in \eqref{invinnerprod}.
\end{Definition}

\begin{Remark}
We will show later that the invariant space is empty for all $+$ or all $-$ representations.
\end{Remark}

}

\begin{Lemma}\label{lm:well-defined}
{The functional $[\Psi]$ } is well-defined if there are at least two representations. 
\end{Lemma}

\begin{proof}
We need to show that for every 
\begin{equation}
\Psi,\Phi\in \bigotimes_i^\wedge S_{j_i^+}^+\hat{\otimes} \bigotimes_i^\wedge S_{j_i^-}^-
\end{equation}
The integral
\begin{equation}
    \int_{\rm SU(1,1)} dg\ \lansu \Phi,g\act \Psi\ransu<\infty
\end{equation}
It is enough to check it on simple tensors (as algebraic tensor product consists of finite sums of simple tensors)
\begin{equation}
    \Psi=\psi_{j_1^+}^+\otimes \ldots\otimes \psi_{j_m^-}^-,\quad
    \Phi=\phi_{j_1^+}^+\otimes \ldots\otimes \phi_{j_m^-}^-,
\end{equation}
We now notice that the integral takes the form
\begin{equation}
    \int_{\rm SU(1,1)} dg\ \prod_{i=1}^n\lansu \phi_{j_i^+}^+,g\act \psi_{j_i^+}^+\ransu \prod_{i=1}^m\lansu \phi_{j_i^-}^-,g\act \psi_{j_i^-}^-\ransu
\end{equation}
which is a product of $\lansu \phi_{j_i^\pm}^\pm,g\act \psi_{j_i^\pm}^\pm\ransu$. Each factor belongs to $L^2(\mathrm{SU(1,1)})$ (see \eqref{eq:orthogonality}), however 
\begin{equation}
    \lt|\lansu \phi_{j_i^\pm}^\pm,g\act \psi_{j_i^\pm}^\pm\ransu\rt|\leq \lt\|\phi_{j_i^\pm}^\pm\rt\|_{\rm SU(1,1)}\,\lt\|\psi_{j_i^\pm}^\pm\rt\|_{\rm SU(1,1)},
\end{equation}
where $\|\cdot\|_{\rm SU(1,1)}$ denote the norm in the hermitian scalar product of the representation. Thus, it is also in $L^\infty(\mathrm{SU(1,1)})$. A product of two $L^2$ functions and many bounded functions belongs to $L^1(\mathrm{SU(1,1)})$. Thus the integral is finite.
\end{proof}

As the integral in case of SU(1,1) contains a $U(1)$ subgroup generated by $L_3$, we see that the integral is zero if the total magnetic number is non-vanishing. This is necessary the case if all representations are $D^+_{j_i}$ or all are $D^-_{j_i}$. We can thus always assume that there is at least one $+$ and one $-$ representation.

We define the operator $P$, which is an SU(1,1) analog of the projection onto the invariant subspace of SU(2):
\begin{equation}\label{eq:P}
    P\colon \bigotimes_i^\wedge S_{j_i^+}^+\hat{\otimes} \bigotimes_i^\wedge S_{j_i^-}^-\rightarrow \operatorname{Inv}\left(\bigotimes_i
    D_{j_i^+}^+\otimes \bigotimes_i D_{j_i^-}^-\right),\quad P\Psi=[\Psi]
\end{equation}
$[\Psi]$ is not inside the original tensor product representation space of SU(1,1), since the naive inner product diverges 
\be
\langle[\Psi],[\Phi]\rangle_{\rm SU(1,1)}&=&\int dgdg'\langle g\act\Psi,g'\act\Phi\rangle_{\rm SU(1,1)}=\int dgdg'\langle \Psi,g^{-1}g'\act\Phi\rangle_{\rm SU(1,1)}\nonumber\\
&=&\int dgdg'\langle \Psi,g'\act\Phi\rangle_{\rm SU(1,1)}=[\Psi](\Phi)\int dg
\ee
due to the infinite Haar volume of SU(1,1). In the 2nd step, we use the invariance of the inner product, then we use the property of the Haar measure in the 3rd step.

\subsection{Orthonormal basis}

Let's consider the simple case $D_j^+\otimes D_j^-$, we define
    \begin{equation}
        \epsilon_j=\lt[\iota_j\rt],\quad \iota_j={\sqrt{d_j}}|j,j\rangle_+\otimes |j,-j\rangle_-
    \end{equation}
where the square bracket maps the state to the invariant. From orthogonality relation we see that for any probe state $\Phi=\phi^+\otimes \phi^-$ and $\Psi=\psi^+\otimes \psi^-$
  \begin{equation}\label{eq:epsilon-def}
        \epsilon_j(\Phi)=\frac{1}{\sqrt{d_j}}\lansu \phi^+,J(\phi^-)\ransu,\quad \epsilon_j^\dagger(\Psi)=\frac{1}{\sqrt{d_j}}\lansu J(\psi^+),\psi^-\ransu,
    \end{equation}
which can be seen by plugging $\iota_j$ and $\Phi, \Psi$ into orthogonality relation \eqref{eq:orthogonality}.
{Namely,
\begin{equation}\label{eq:orthogonality-2}
    \int_{\rm SU(1,1)} dg\ \lansu \psi^+,g\act \phi^+\ransu
    \lansu \psi^-,g\act \phi^-\ransu=
    \epsilon_j(\psi^+\otimes \psi^-)
    \epsilon_j^\dagger(\phi^+\otimes \phi^-)
\end{equation}
}
We can interpret the orthogonality relation as equality of operators $S_j^+\hat{\otimes}S_j^-\rightarrow \operatorname{Inv}(D_j^+\otimes D_j^-)$ 
    \begin{equation}
        P:=P_j=\epsilon_j \epsilon_j^\dagger.
    \end{equation}
This follows from \eqref{eq:orthogonality-2}
    \begin{equation}
    [\Psi](\Phi)=\epsilon_j(\Phi)\epsilon_j^\dagger(\Psi).
    \end{equation}
For $D_j^+\otimes D_{j'}^-$ with $j\neq j'$, the invariant subspace $\operatorname{Inv}(D_j^+\otimes D_{j'}^-)$ is empty due to the orthogonality condition \eqref{eq:orthogonality}.

For any $k, k'$ such that $j^+(k)=j^-(k')=j$ where $j^\pm(k)$ were defined in \eqref{eq:D-decomp}, we consider
    \begin{equation}
        \ci_{kk'}=[i_{kk'}],\quad i_{kk'}= {\sqrt{d_j}}{I^+_k}^\dagger|j,j\rangle_+\otimes {I^-_{k'}}^\dagger|j,-j\rangle_-\label{IjjIjj}.
    \end{equation}    
    {It is a well-defined invariant because ${I^\pm_k}^\dagger$ map  elements of the canonical basis into $S$.}
    Using intertwining property of $I^\pm_k$ we have
    \begin{equation}
        \ci_{kk'}={I^+_k}^\dagger\otimes {I^-_{k'}}^\dagger\epsilon_j
    \end{equation}

\begin{Lemma}\label{lm:inv}
As an operator \eqref{eq:P} we have
\begin{equation}\label{PII}
    P=\sum_{k,k'\colon j^+(k)=j^-(k')}\ci_{kk'} \ci_{kk'}^\dagger
\end{equation}
\end{Lemma}

\begin{proof}
{
We need to show that
\begin{equation}\label{eq:P-resolved}
    [\Phi](\Psi)=\sum_{k,k'\colon j^+(k)=j^-(k')}\ci_{kk'}(\Psi) \ci_{kk'}^\dagger(\Phi).
\end{equation}
The spaces $S$ consist of finite linear combinations of simple tensors thus it is enough to check the identity for vectors. We write
\begin{equation}
    \Phi=\Phi^+\otimes \Phi^-,\quad\Psi=\Psi^+\otimes \Psi^-,
\end{equation}
where $\Phi^\pm,\Psi^\pm\in \hat{\bigotimes}_i S_{j_i^\pm}^\pm$.
The following identity holds
\begin{equation}
    \left\langle \Psi^\pm,g\act \Phi^\pm\right\rangle_{SU(1,1)}=
    \sum_{k}
    \left\langle I_k^\pm(\Psi^\pm),g\act I_{k}^\pm(\Phi^\pm)\right\rangle_{SU(1,1)}
\end{equation}
The sum on the left-hand side might be infinite, but it converges to the right-hand side in $L^2(SU(1,1))$. We obtain
\begin{align}
    &\langle \Psi^+,g\act \Phi^+\rangle_{SU(1,1)}\langle \Psi^-,g\act \Phi^-\rangle_{SU(1,1)}\nonumber\\
    &=
    \sum_{k,k'}
    \left\langle I_k^+(\Psi^+),g\act I_{k}^+(\Phi^+)\right\rangle_{SU(1,1)}\ 
    \left\langle I_{k'}^-(\Psi^-),g\act I_{k'}^-(\Phi^-)\right\rangle_{SU(1,1)}\label{eq:decomp}
\end{align}
The equality is in the sense of $L^1(SU(1,1))$ as a product of two square integrable functions. The equations \eqref{eq:orthogonality-2} show that
\begin{align}
    &\int dg\ 
    \left\langle I_k^+(\Psi^+),g\act I_{k}^+(\Phi^+)\right\rangle_{SU(1,1)}\ 
    \left\langle I_{k'}^-(\Psi^-),g\act I_{k'}^-(\Phi^-)\right\rangle_{SU(1,1)}\nonumber\\
    &=\delta_{j^+(k)=j^-(k')}
   \epsilon_k (I_k^+(\Psi^+)\otimes I_{k'}^-(\Psi^-))\ 
      \epsilon_k^\dagger (I_k^+(\Phi^+)\otimes I_{k'}^-(\Phi^-))
\end{align}
After substituting into \eqref{eq:decomp} we obtain the identity \eqref{eq:P-resolved}.
}

\end{proof}

{
\begin{Remark}
The Lemma \ref{lm:inv} shows that the scalar product is positively defined as
\begin{equation}
    \laninv [\Phi],[\Phi]\raninv=P(\Phi)(\Phi)=
    \sum_{k,k'\colon j^+(k)=j^-(k')}|\ci_{kk'}(\Phi)|^2\geq 0.
\end{equation}
\end{Remark}
}

The sum in $P$ in \eqref{PII} is not always infinite. Recall our terminology in Remark \ref{remarktimelike}. For examples, timelike tetrahedra with 3 face-normals future-pointing and 1 past-pointing correspond to invariants in $D_{j_1}^+\otimes D_{j_2}^+\otimes D_{j_3}^+\otimes D_{j_4}^-$. The sum in \eqref{PII} becomes a finite sum over $k$ due to the constraint $j^+(k)=j_4^-$ where $j^+(k)$ labels subspaces in $D_{j_1}^+\otimes D_{j_2}^+\otimes D_{j_3}^+$. The space of intertwiners is finite-dimensional. Similar argument holds for tetrahedra with 3 face-normals past-pointing and 1 future-pointing, and generalizes to timelike polyhedra with only one face-normal future-pointing or only one face-normal past-pointing. However for all other cases, e.g. polyhedra with 2 or more face-normals future-pointing and past-pointing, the sum in \eqref{PII} extends to infinity.

Recall the inner product on the space of SU(1,1) intertwiners \eqref{invinnerprod}, 
we compute 
\begin{align}
    \laninv \ci_{kk'},\ci_{ll'}\raninv=
    &\int_{\rm SU(1,1)} dg \, {\sqrt{d_{j^+(k)}d_{j^+(l)}}}\lansu j^+(k),j^+(k)|I_k^+{I_l^+}^\dagger g\act |j^+(l),j^+(l)\ransu\nonumber\\ &\lansu j^-(k'),-j^-(k')|I_k^-{I_l^-}^\dagger g\act |j^-(l'),-j^-(l')\ransu
\end{align}
However $j^+(k)=j^-(k')$, $j^+(l)=j^-(l')$ and
\begin{equation}
    I_k^\pm {I_l^\pm}^\dagger =\delta_{kl}{\mathbb{I}}
\end{equation}
thus denoting $j^+(k)=j^+(l)=j$
\begin{equation}
    \laninv \ci_{kk'},\ci_{ll'}\raninv= {d_j}
    \delta_{kl}\delta_{k'l'} \int_{\rm SU(1,1)} dg\ \lansu j,j|g\act |j,j\ransu \lansu j,-j|g\act |j,-j\ransu=\delta_{kl}\delta_{k'l'}
\end{equation}
from orthogonality of Wigner matrices. The orthogonality of $\ci_{kk'}$ and \eqref{PII} shows that $\ci_{kk'}$ form a complete orthonormal basis in the intertwiner space.

\subsection{Duality}

{
Our definition of the spin foam model involves orientations  of edges and faces. We were considering so far the intertwiner space in the case when all these orientations coincides. In order to extend it to a more general circumstances we need to introduce duality maps.
}

$\zeta:=\sqrt{d_j}\epsilon_j^\dagger\in \mathrm{Inv}(D_j^+\otimes D^-_j)$ defines a linear isomorphism from $D_j^-$ to the dual space $D_j^+{}^*$, and from $D_j^+$ to $D_j^-{}^*$ (there is an anti-linear isomorphism between $D_j^\pm{}^*$ and $D_j^\pm$ by Riesz's theorem.)
\be
\zeta:&& D_j^-\to D_j^+{}^*\\
&&\psi^-\mapsto \langle J(\psi^-) ,\ \cdot\  \rangle_{\rm SU(1,1)},\quad \forall\ \psi^-\in D_j^-\\
\zeta:&& D_j^+\to D_j^-{}^*\\
&&\psi^+\mapsto \langle J(\psi^+),\ \cdot\  \rangle_{\rm SU(1,1)},\quad \forall\ \psi^+\in D_j^+
\ee
The unitarity and intertwining property of $\zeta$ follows from the anti-unitarity and intertwining property of $J$. $J$ maps any finite linear combination of the canonical basis in $ D_j^\pm$ to a finite linear combination of the canonical basis in $ D_j^\mp$.


With $\zeta$, we define the following map  
\be
\zeta^{(i)}:\ [\Phi]\mapsto\widetilde{[\Phi]}:=[ \zeta_i\Phi]=\int_{\rm SU(1,1)}dg\,\langle \,\cdot\,,\,g\act \zeta_i\Phi\rangle_{\rm SU(1,1)},
\ee
where $\zeta_i=1\otimes\cdots\otimes\zeta\otimes\cdots\otimes1$ with $\zeta$ appears only at the $i$-th slot. The inner product inside the integral has been generalized to include $D_{j_i}^\pm{}^*$. 

\begin{Definition}

The intertwiner space with $D_{j_i}^-\to D_{j_i}^+{}^*$ or $D_{j_i}^+\to D_{j_i}^-{}^*$ at the $i$-th factor is the linear span of the image $\zeta^{(i)}{[\Phi]}=[ \zeta_i\Phi]$ completed by the following inner product 
\be
\langle [\zeta_i\Psi],[\zeta_i\Phi]\rangle_{\rm Inv}&:=& [\zeta_i\Phi](\zeta_i\Psi)=\int_{\rm SU(1,1)}dg\langle \zeta_i\Psi,\,g\act \zeta_i\Phi\rangle \\
&=&\int_{\rm SU(1,1)}dg\langle \Psi,\,g\act \Phi\rangle=\langle [\Psi],[\Phi]\rangle_{\rm Inv}
\ee 
\end{Definition}

For instance, if $[\Phi]\in \mathrm{Inv}(D_{j_1}^+\otimes D_{j_2}^+\otimes D_{j_3}^+\otimes D_{j_4}^-)$ the image of $\zeta^{(4)}$ spans $\mathrm{Inv}(D_{j_1}^+\otimes D_{j_2}^+\otimes D_{j_3}^+\otimes D_{j_4}^+{}^*)$. $\zeta^{(i)}$ is an isomorphism by definition. Applying $\zeta^{(i)}$ to the orthonormal basis: $\zeta^{(i)}\ci_{kk'}=[\zeta_i i_{kk'}]$, we have 
\be
\langle \zeta^{(i)}\ci_{kk'},\zeta^{(i)}\ci_{ll'}\rangle_{\rm Inv}=\delta_{kl}\delta_{k'l'}.
\ee 
i.e. $\zeta^{(i)}\ci_{kk'}$ is an orthonormal basis.

Composing $\zeta^{(i)}$ defines isomorphisms from intertwiner spaces with $D_j^\pm$ to intertwiner spaces with one or several $D_j^\pm$ changed to $D_j^\mp{}^*$. These isomorphisms translate the projector $P$ and the orthonormal basis $\ci_{kk'}$ to intertwiner spaces with $D_j^\pm{}^*$. In the following, Eq.\eqref{PII} is understood to be on any intertwiner space with or without $D_j^\pm{}^*$, and $\ci_{kk'}$ form an orthonormal basis in the space. 


\subsection{Vertex amplitude and spinfoam amplitude}

The vertex amplitude $A_\Gamma(\{\psi_{\fn\fl}^{\epsilon_{\fn\fl}}\}_{\fn,\fl})$ with boundary states $\{\psi_{\fn \fl}^{\epsilon_{\fn\fl}}\}_{\fn,\fl}$ being either coherent states or the canonical basis is integrable by Corollary \ref{cor:main}. The finiteness extends to finite linear combination of such states. We notice that for any state $A_\Gamma(\{g_{\fn}\act\psi_{\fn\fl}^{\epsilon_{\fn\fl}}\}_{\fn,\fl})=A_\Gamma(\{\psi_{\fn\fl}^{\epsilon_{\fn\fl}}\}_{\fn,\fl})$. In fact, the vertex amplitude $A_\Gamma$ can be understood as a function of intertwiners. 
We define
\be
A_\Gamma(\{\ci_{\fn}\}_\fn):=A_\Gamma(\{i_{\fn}\}_\fn)
\ee
for any $\ci_{\fn}=[i_{\fn}]$. Here $i_\fn$ is $i_{kk'}$ in \eqref{IjjIjj} generalized to be possibly living in dual spaces by the action of $\zeta$. $i_\fn$ contributes factors in $D^\pm_j$ to the kets in \eqref{AGamma}, and contributes factors in $D^\pm_j{}^*$ to bras. Note that the result does not depend on the choice of representative of invariant. As $i_\fn$ is a finite linear combination of tensors of canonical basis, we have 

\begin{Corollary}
$A_\Gamma(\{\ci_{\fn}\}_\fn)$ is finite.
\end{Corollary}

We are now ready to define the spinfoam amplitude including both spacelike and timelike polyhedra. We consider an arbitrary 4d cellular complex $\ck$ and its dual complex $\ck^*$. Objects in $\ck^*$ are vertices $v$, oriented edges $e$, and oriented faces $f$. These objects in $\ck^*$ are used to label spinfoam variables. Firstly, every face $f$ is colored by an SU(2) or SU(1,1) spin $j_f\in \mathbb{N}_0/2$, and $j_f>1/2$ if any $e\subset \partial f$ is dual to a timelike polyhedron in $\ck$. Secondly, every oriented edge $e$ associates with an intertwiner space of and the projector $P_e$ (see also \cite{KKL,Bahr:2010bs})
\be
P_e=\sum_{\ci_e}\ci_e\otimes \ci_e^\dagger,\quad \ci_e\in\mathrm{Inv}\lt(\bigotimes_{f_1} D^{\eps_{f_1}}_{j_{f_1}}\bigotimes_{f_2} D^{\eps_{f_2}}_{j_{f_2}}{}^*\rt),\quad \eps_f=0,\pm.
\ee
Here $\ci_e=[i_e]$ is the orthonormal basis with $i_e=i_{kk'}$ in \eqref{IjjIjj} (generalized to be possibly living in dual spaces by the action of $\zeta$). $f_1$ (or $f_2$) are faces whose orientation inducing the orientation of $\partial f$ is congruent (or conflict) with the orientation of $e$. $\eps_f=\pm$ (or 0) when $e$ is dual to a timelike (or spacelike) tetrahedron. We associate $\ci_e$ to the source $s(e)$ and associate $\ci_e^\dagger$ to the target $t(e)$. $\ci_e$ and $\ci_e^\dagger$ contribute to the vertex amplitudes at $s(e)$ and $t(e)$ respectively.

At each $v$ dual to a 4-simplex $\sig$, the boundary polyhedra of $\sig$ makes a partition of $\partial \sig$. The partition is dual to an oriented graph $\G$. Each link $\fl$ and each node $\fn$ of $\G$ correspond uniquely to a $f$ and $e$ adjacent to $v$ respectively. The orientation of $\fl$ is congruent to the orientation of $f$. We re-label the associated spin by $ j_\fl\equiv j_f$ and the intertwiner basis by $\ci_{\fn}\equiv \ci_{e}$ in $A_\G$. We denote the vertex amplitude by 
\be
A_v\lt(j_f, \ci_{e}\rt)\equiv A_\Gamma\lt(\lt\{\ci_{\fn}\rt\}_\fn\rt)
\ee
Here we also use the same notation $\ci_{e}$ to denote SU(2) intertwiners. The spinfoam amplitude is defined by 
\be
A_\ck=\sum_{\{j_f\}}\prod_{f} A_f({j_f})\, \cz_\ck\lt(\{j_f\}\rt),\quad
\cz_\ck\lt(\{j_f\}\rt)=\sum_{\{\ci_e\}}\prod_vA_v\lt(j_f, \ci_e\rt),\label{Ackapl}
\ee
where $A_f({j_f})$ is an arbitrary face amplitude. $\prod_f$ is over internal faces. $\sum_{\{j_f\}}$ and $\sum_{\{\ci_e\}}$ are over internal spins and intertwiners of internal edges.

The above definition of $A_\ck$ is generally formal because (1) the sum over $j_f$ can cause divergence, and (2) even for $\cz_{\ck}$ with fixed $j_f$'s, the sum over SU(1,1) intertwiners is not obviously finite, when the intertwiner space is infinite-dimensional.

If every internal polyhedron has only one face-normal future-pointing or only one face-normal past-pointing, the sum $\sum_{\ci_e}$ is finite, then $\cz_\ck\lt(\{j_f\}\rt)$ is well-defined, and is independent of the choice of the particular intertwiner basis.

\section{Spinfoam amplitude with additional cut-off}\label{spinfoam amplitude with additional cut-off}

We are not able to prove the finiteness of $\cz_\ck\lt(\{j_f\}\rt)$ in the presence of the timelike polyhedron with at least 2 future-pointing and 2 past-pointing face-normals, and we suspect that the finiteness may not be true in this case. In the following, we explain the situation firstly at the semiclassical level in the case of tetrahedron, then we propose the additional cut-off that we need to impose on the sum over SU(1,1) intertwiners in the amplitude, in case that the sum diverges. The restriction has very clear semiclassical geometric picture, which we explain in Section \ref{sec:shadows}. By the restriction, the finiteness of $\cz_\ck\lt(\{j_f\}\rt)$ is proven straightforwardly and will be explained at the end of this section.

\subsection{Semiclassical non-compactness}\label{Necessity}

The semiclassical motivation of the additional cut-off in the space of intertwiners can be understood from the following facts: There exists arbitrary large tetrahedra even with fixed areas and volume. This phenomena only occur when two face-normals are future-pointing and two are past-pointing.

{We exhibit this fact on an example with null faces. The example with spacelike faces can be obtained by a small perturbation.} Let us define face vectors ($\vec{e}_0$ is timelike, $\vec{e}_1$ and $\vec{e}_2$ unit spacelike an orthonormal basis)
\begin{equation}
    \vec{n}_1=\alpha(\vec{e}_0+\vec{e}_1),\quad \vec{n}_2=\alpha(\vec{e}_0-\vec{e}_1),\quad \vec{n}_3=-U(\phi)\vec{n}_1,\quad \vec{n}_4=-U(\phi)\vec{n}_2,
\end{equation}
where $U(\phi)$ is rotation around $e_0$. We have $\sum_i n_i=0$.
We can reconstruct tetrahedron with the face bivectors given by dual of the vectors. The volume is proportional to $\alpha^{3/2}(\phi+O(\phi^2))$. We can choose $\phi=\frac{1}{\alpha^4}$ and $\alpha\rightarrow\infty$. The volume is very small but the bivectors are large (although they are null).

In the case of two face-normals future-pointing and two past-pointing, the size of faces nor the volume do not restrict the shapes of tetrahedra to a compact space. As an example, we consider a timelike tetrahedron with all faces having equal area $A_1=A_2=A_3=A_4=1$, two face-normals future-pointing and two past-pointing. The tetrahedron volume is given by (see Appendix \ref{Classical tetrahedron} for some details) 
\be
V=\frac{2}{3} \sqrt[4]{2} \sqrt[4]{\cosh ^{2}\left(\frac{\theta_{12}}{2}\right) \sinh ^{2}\left(\frac{\theta_{13}}{2}\right)\left(\cosh \left(\theta_{12}\right)-\cosh \left(\theta_{13}\right)\right)}
\ee
where real variables $\theta_{12}$ and $\theta_{13}$ are dihedral boost angles between faces 1,2 and 1,3. There is a noncompact space of $\theta_{12},\theta_{13}$ which give a fixed $V$, since increasing $\theta_{12}$ and $\theta_{13}$ can leave $V$ invariant provided that $(\cosh \left(\theta_{12}\right)-\cosh \left(\theta_{13}\right)$ decreases. In contrast, when there are 1 face-normal future-pointing (past-pointing) and 3 face-normals past-pointing (future-pointing), the space of $\theta_{12},\theta_{13}$ is compact (see Appendix \ref{Classical tetrahedron} for a proof).

The above semiclassical intuition indicates that the amplitude is possible to diverge in this case. Unfortunately we do not have a proof for either the finiteness or the divergence of the amplitude in presence of internal tetrahedron with two future-pointing and two past-pointing face-normals. In the following, we assume the worst scenario that the amplitude is divergent in this case, and we develop a scheme of cut-off with natural quantum geometrical meaning in order to render the amplitude finite.

\subsection{Shadows}\label{sec:shadows}

We investigate the geometric notion of a shadow: Consider a timelike polytope (convex polyhedron) whose oriented area vector of the faces are denoted by $\vec{n}_i^\pm$. Let  $\vec{n}_i^+$ be timelike future and $\vec{n}_i^-$ timelike past directed vectors, and we assume that they satisfy
\begin{equation}\label{eq:polytope-vectors}
    \sum_{i=1}^n\vec{n}_i^++\sum_{i=1}^m\vec{n}_i^-=0
\end{equation}
We denote the Minkowski squared norm by $\|\vec{u}\|^2$ and it is positive for timelike vectors $\vec{u}$. We denote $\|\vec{v}\|=\sqrt{\|\vec{v}\|^2}$
for the timelike vector $\vec{v}$. We introduce 
\begin{equation}
 \vec{N}=\sum_{i=1}^n\vec{n}_i^+=-\sum_{i=1}^m\vec{n}_i^-   
\end{equation}
and $\lambda=\|\vec{N}\|$. 

We call a shadow the extremal area of a polytope as seen from any direction. More precisely, for any choice of subsets $\Fu_+\subset\{1,\ldots n\}$, $\Fu_-\subset \{1,\ldots m\}$, we define the sum
\begin{equation}
    \vec{s}=\vec{s}_++\vec{s}_-,\quad \vec{s}_\pm=\sum_{i\in \Fu_\pm}\vec{n}_i^\pm,
\end{equation}
{The area of the polyhedron as seen from a direction $\vec{m}$ (where $\|\vec{m}\|^2=\pm 1$ depending whether it is timelike or spacelike) is equal to
\begin{equation}
    |\vec{m}\cdot\vec{s}|
\end{equation}
where $\Fu_\pm$ contain only visible faces. This leads to the following definition of a shadow:
The shadow is defined by the formula
\begin{equation}
    \mathrm{Max}_{\Fu_\pm}(\sqrt{|\|\vec{s}\|^2|}).
\end{equation} 
}

We have the following result:

\begin{Proposition}
The shadow of the polytope with face vectors given by \eqref{eq:polytope-vectors} is equal to $\lambda$.
\end{Proposition}

\begin{proof}
{We know that $\|\vec{N}\|=\lambda$ thus the shadow is bounded from below by $\lambda$. We will show that for any choice of subsets $\Fu_+\subset\{1,\ldots n\}$, $\Fu_-\subset \{1,\ldots m\}$, $\vec{s}$ satisfies 
\begin{equation}
    -\lambda^2\leq \|\vec{s}\|^2\leq \lambda^2.
\end{equation}}
We denote $\vec{N}=\lambda \vec{e}_0$ ($\vec{e}_0$ unit future directed) and decompose
\begin{equation}
    \vec{s}_\pm=\lambda_\pm \vec{e}_0+\vec{s}_\pm ^\perp,\quad \lambda_\pm=\vec{s}_\pm\cdot \vec{e}_0,
\end{equation}
and $\lambda_+\geq 0$, $\lambda_-\leq 0$.
The space $\vec{e}_0^\perp$ is Euclidean and we introduce in it a norm
\begin{equation}
    \|\vec{v}\|_\perp^2=-\|\vec{v}\|^2,\quad \vec{v}\cdot\vec{e}_0=0.
\end{equation}
This means that
\begin{equation}
\|\vec{s}_\pm\|^2=(\vec{s}_\pm\cdot\vec{e}_0)^2-\|\vec{s}_\pm^\perp\|_\perp^2,\quad \|\vec{s}\|^2=(\vec{s}\cdot\vec{e}_0)^2-\|\vec{s}^\perp\|_\perp^2.
\end{equation}
Let us notice that both {$\pm\vec{s}_\pm$ and $\vec{N}\mp\vec{s}_\pm$ are timelike and future directed
\begin{equation}
\pm \vec{s}_\pm=\pm \sum_{i\in \Fu_\pm}\vec{n}_i^\pm,\quad
    \vec{N}\mp\vec{s}_\pm=\pm \sum_{i\notin \Fu_\pm}\vec{n}_i^\pm,
\end{equation}} 
thus
\begin{equation}
    \|\vec{s}_\pm\|^2=\lambda_\pm^2-\|\vec{s}_\pm^\perp\|_\perp^2\geq 0,\quad (\vec{N}\mp \vec{s}_\pm)\cdot\vec{e}_0=\lambda\mp\lambda_\pm\geq 0,\quad\|\vec{s}_\pm\mp \vec{N}\|^2=(\lambda\mp\lambda_\pm)^2-\|\vec{s}_\pm^\perp\|_\perp^2\geq 0,
\end{equation}
Thus we obtain
\begin{equation}
    \|\vec{s}_\pm^\perp\|_\perp\leq \min(\pm\lambda_\pm,\lambda\mp\lambda_\pm)\leq \frac{\pm\lambda_\pm+(\lambda\mp\lambda_\pm)}{2}\lambda=\frac{1}{2}\lambda
\end{equation}
and estimating from below
\begin{equation}
\|\vec{s}\|^2\geq -\|\vec{s}_+^\perp+\vec{s}_-^\perp\|^2_\perp\geq -(\|\vec{s}_+^\perp\|_\perp+\|\vec{s}_-^\perp\|_\perp)^2\geq -\lambda^2
\end{equation}
Similarly,
\begin{equation}
    \|\vec{s}\|^2\leq (\lambda_++\lambda_-)^2\leq \max(|\lambda_+|,|\lambda_-|)^2\leq \lambda^2
\end{equation}
Geometric interpretation is that $\vec{s}$ is possible sum of visible faces in polytope, thus the visible surface (shadow) cast by the polytope is bounded by $\lambda$. This finishes the proof of the proposition.
\end{proof}

\subsection{Quantum cut-off}

Given the spinfoam amplitude with fixed spins at all faces, we introduce the cut-off $j^+(k)\leq M$ for every timelike polyhedra where the sum in \eqref{PII} is infinite. This means that we restrict eigenvalues of the ``shadow operator''
\begin{equation}
    \left(\sum_{i=1}^n \hat{\bm X}_i^+\right)^2,\quad \hat{\bm X}_i^+=\lt(\hat{L}^3,\hat{K}^1,\hat{K}^2\rt)_i^+
\end{equation}
{by $M(M-1)$}. Here the superscript $+$ indicates that the SU(1,1) generators are represented on $D^+_j$. This cut-off is translated to the cut-offs of the intertwiner spaces with $D^\pm_j{}^*$ by the isometries defined by $\zeta$.

We modify $\cz_{\ck}(\{j_f\})$ by imposing the cut-off to the sum of intertwiners so that $\sum_{\ci_e}$ becomes a finite sum. The resulting $\cz_{\ck}(\{j_f\})$ is finite.

\begin{Corollary}
The spinfoam amplitude $A_\ck$ in \eqref{Ackapl} on any cellular complex $\ck$ is finite when 2 types of cut-offs are imposed: one is imposed on the area quantum number $j_f\leq j_{max}$ (in the case of the bubble divergence), the other is imposed on the shadow quantum number $j^+(k)\leq M$ for every internal timelike polyhedron that has at least 2 future-pointing and 2 past-pointing face-normals. 
\end{Corollary}

\section{Conclusion and outlook}\label{sec4}

In summary, we firstly prove the finiteness of the extended spinfoam vertex amplitude in presence of timelike polyhedra. When generalizing the proof of finiteness to the extended spinfoam amplitude $A_\ck$ on the cellular complex $\ck$, we show that $A_\ck$ is finite once we impose cut-offs on eigenvalues of area operators (of internal faces) and shadow operators (of internal polyhedra).

In this paper, we focus on the case of spinfoam amplitude with timelike polyhedra whereas all faces are still spacelike. The finiteness of the amplitude with timelike faces is not addressed in this paper. The analysis of amplitudes with timelike faces is expected to be more involved, because it relates to the continuous-series unitary irrep of SU(1,1), then neither the canonical basis nor coherent boundary states of vertex amplitudes are normalizable. Therefore we postpone this analysis to the future publication. 

There are a few other future aspects that we would like to mention: Recently various numerical techniques have been applied to the EPRL spinfoam model with spacelike tetrahedra \cite{Han:2020npv,Dona:2019dkf}. It is interesting to generalize these numerical methods to include the extended model with timelike tetrahedra/polyhedra, see \cite{github} for the first step. It is also important to consider the deformation of the extended spinfoam model to include cosmological constant, since so far the cosmological constant is only implemented to the model with spacelike tetrahedra \cite{QSF,QSF1,HHKR,Han:2021tzw}. There is evidence suggesting that the inclusion of cosmological constaint might not only introduce a natural cut-off of spins in $A_\ck$, but also make the sum over SU(1,1) intertwiner convergent. Lastly, the inclusion of other matter fields in the extended spinfoam theory should also be an interesting future perspective.


\section*{Acknowledgements}

M.H. receives support from the National Science Foundation through grant PHY-1912278.  W.K acknowledges grant of Polish National Science Centre (NCN) Sheng-1 2018/30/Q/ST2/00811.
	

\appendix

\section{Coherent states}\label{Coherent state}

The SU(1,1) discrete-series unitary irrep can be realized by representation functions on SU(1,1). The orthonormal basis $|j,m\rangle^\pm$ in $\cd^\pm_j$ is realized by
\be
\begin{array}{l}
\psi_{jm}^{+}(v)=\sqrt{2j-1}D_{jm}^{+,j}(v),\quad m\geq j\\
\psi_{jm}^{-}(v)=\sqrt{2j-1}D_{-jm}^{-,j}(v),\quad m\leq-j,\quad v\in\mathrm{SU(1,1)}
\end{array}
\ee
where $D^{\pm, j}_{nm}(v)$ is the SU(1,1) representation matrix. Similarly the orthonormal basis in the SU(2) unitary irrep can be relatized by 
\be
\psi_{jm}(v)=\sqrt{2j+1}D^j_{jm}(v),\quad v\in\Su
\ee
where $D^j_{nm}(v)$ is the SU(2) representation matrix (Wigner $D$-matrix).

The coherent states $|j,\mathbf{m}^\eps\rangle^\eps$ is realized by 
\be
\psi^{\pm,\mathbf{m}^\pm}_j(v)&=&\psi^\pm_{jj}(vu),\quad \mathbf{n}^{\pm}=\left(u^{T}\right)^{-1}\mathbf{n}_{\pm},\quad u\in\mathrm{SU(1,1)},\\
\psi^{\mathbf{n}}_j(v)&=&\psi_{jj}(vu),\quad\mathbf{n}=\left(u^{T}\right)^{-1}\mathbf{n}_{+},\quad u\in\Su.
\ee
where
\be
\mathbf{n}_{+}=\left(\begin{array}{l}
1\\
0
\end{array}\right)
\qquad
\mathbf{n}_{-}=\left(\begin{array}{l}
0\\
1
\end{array}\right)
\ee
The spinfoam $Y$-map embeds SU(1,1) or SU(2) states into the $\Slc$ unitary irrep $(\rho,n)=(2\g j,2j)$. $Y_{\pm}$ acting on the SU(1,1) state $\psi^\pm_{jm}$ gives
\be
Y_\eps \psi^\eps_{jm}\equiv F_{jm}^{\eps}(\mathbf{z})=\frac{1}{\sqrt{\pi}}\Theta\left(\eps\langle\mathbf{z}\mid\mathbf{z}\rangle\right)\left(\eps\langle\mathbf{z}\mid\mathbf{z}\rangle\right)^{\mathrm{i}\frac{\rho}{2}-1}\psi_{jm}^{\eps}\left(v^{\eps}(\mathbf{z})\right),\quad
\ee
where $\eps=\pm$ and
\be\label{vz1}
&v^{+}(\mathbf{z})=\frac{1}{\sqrt{\langle\mathbf{z}\mid\mathbf{z}\rangle}}\left(\begin{array}{ll}
z_{+} & z_{-}\\
\bar{z}_{-} & \bar{z}_{+}
\end{array}\right),\quad v^{-}(\mathbf{z})=\frac{1}{\sqrt{-\langle\mathbf{z}\mid\mathbf{z}\rangle}}\left(\begin{array}{cc}
\bar{z}_{-} & \bar{z}_{+}\\
z_{+} & z_{-}
\end{array}\right),\quad\mathbf{z}=\left(\begin{array}{l}
z_{+}\\
z_{-}
\end{array}\right)&\\
&\langle\mathbf{z}\mid\mathbf{z}'\rangle:=\bar{z}_{+}z'_{+}-\bar{z}_{-}z'_{-}&
\ee
$Y_0$ acting on the SU(2) state $\psi_{jm}$ gives
\be
Y_0 \psi_{jm}\equiv F_{jm}(\mathbf{z})=\frac{1}{\sqrt{\pi}}\langle\mathbf{z} \mid \mathbf{z}\rangle_{0}^{\mathrm{i} \frac{\rho}{2}-1} \psi_{j m}(v(\mathbf{z}))
\ee
where
\be\label{vz2}
&v(\mathbf{z})=\frac{1}{\sqrt{\langle\mathbf{z} \mid \mathbf{z}\rangle_{0}}}\left(\begin{array}{cc}
z_{+} & z_{-} \\
-\bar{z}_{-} & \bar{z}_{+}
\end{array}\right), \quad \mathbf{z}=\left(\begin{array}{c}
z_{+} \\
z_{-}
\end{array}\right)&\\
&\langle\mathbf{z}\mid\mathbf{z}'\rangle_{0}:=\bar{z}_{+}z'_{+}+\bar{z}_{-}z'_{-}&
\ee
$\langle\cdot\mid\cdot\rangle$ and $\langle\cdot\mid\cdot\rangle_0$ are SU(1,1) and SU(2) invariant inner products on $\C^2$. 

The $Y$-map acting on highest (or lowest) weight states gives
\be
F_{jj}^{+}(\mathbf{z})&=&\sqrt{\frac{2j-1}{\pi}}\Theta\left(\langle\mathbf{z}\mid\mathbf{z}\rangle\right)\langle\mathbf{z}\mid\mathbf{z}\rangle^{\mathrm{i}\frac{\rho}{2}-1+j}\left\langle \mathbf{z}\mid\mathbf{n}_{+}\right\rangle ^{-2j}, \label{fjj1}\\
F_{j-j}^{-}(\mathbf{z})&=&\sqrt{\frac{2j-1}{\pi}}\Theta\left(-\langle\mathbf{z}\mid\mathbf{z}\rangle\right)\left(-\langle\mathbf{z}\mid\mathbf{z}\rangle\right)^{\mathrm{i}\frac{\rho}{2}-1+j}\left(-\left\langle \mathbf{z}\mid\mathbf{n}_{-}\right\rangle \right)^{-2j},\label{fjj2}\\
F_{jj}(\mathbf{z})&=&\sqrt{\frac{2 j+1}{2 \pi}}\langle\mathbf{z} \mid \mathbf{z}\rangle_{0}^{\mathrm{i} \frac{\rho}{2}-1-j}\left\langle \mathbf{n}_{+}\mid  \mathbf{z} \right\rangle_{0}^{2 j}.\label{fjj3}
\ee
The $Y$-map acting on the coherent state $|j,\mathbf{n}^\eps\rangle$ is realized by
\be
\Psi_j^{\pm,\mathbf{m}^{\pm}}(\mathbf{z})&=&\Theta\left(\pm\langle\mathbf{z}\mid\mathbf{z}\rangle\right)\sqrt{\frac{2j-1}{\pi}}(\pm\langle\mathbf{z}\mid\mathbf{z}\rangle)^{\mathrm{i}\frac{\rho}{2}-1+j}(\pm\left\langle \mathbf{z}\mid\mathbf{m}^{\pm}\right\rangle) ^{-2j},\quad\mathbf{m}^{\pm}=\left(u^{T}\right)^{-1}\mathbf{n}_{\pm},\label{app:coh1}\\
\Psi_j^{\mathbf{m}}(\mathbf{z})&=&\sqrt{\frac{2 j+1}{2 \pi}}\langle\mathbf{z} \mid \mathbf{z}\rangle_{0}^{\mathrm{i} \frac{\rho}{2}-1-j}\langle\mathbf{m}\mid \mathbf{z} \rangle_{0}^{2 j},\quad \quad\mathbf{m}=\left(u^{T}\right)^{-1}\mathbf{n}_{+}.\label{app:coh2}
\ee

\section{Comparison of \texorpdfstring{$\beta$}{beta} and the scalar product}
\label{sec:comparison}

\subsection{Properties of the $\hat{\b}$-map}

From intertwining property
\be
\hat{\beta}(g\act \Phi)=g\act \hat{\beta}(\Phi)
\ee
If we take $g(t)$ of the one parameter subgroup then the generator is defined by
\be
b\Psi=\left.\frac{1}{i}\frac{d(g(t)\act \Psi)}{dt}\right|_{t=0}
\ee
as the map is antilinear we get
\be
\hat{\beta}(b\Phi)=-b\hat{\beta}(\Phi)
\ee
In particular 
\be
\label{eq:beta-interntwiner}
\hat{\beta}(L^2\Phi)=L^2\hat{\beta}(\Phi),\quad 
\hat{\beta}(L_3\Phi)=-L_3\hat{\beta}(\Phi),
\ee
where depending on the case of coherent state $L^2$ is a Casimir of $\Su$ or $SU(1,1)$.

\begin{Lemma}
For any $s\in \{-,+\}$ there exists a unique state in $\ch^{(\rho,n)}$ up to multiplication by complex number that satisfies
\begin{enumerate}
    \item $L^2_{SU(1,1)}\Phi=\frac{1}{4}n(n-2)\Phi$
    \item {$L_3\Phi=s\frac{1}{2}m\Phi$, $m\geq n$}
\end{enumerate}
and it is equal to {$F_{j,sm/2}^{s}$ for $j=n/2$ (in case of $n=m$ this is just $\Psi_j^{s,\mathbf{n}_s}$)}. 
There exists a unique state in $\ch^{(\rho,n)}$ up to multiplication by complex number that satisfies
\begin{enumerate}
    \item $L^2_{\Su}\Phi=\frac{1}{4}n(n+2)\Phi$
    \item {$L_3\Phi=\frac{1}{2}m\Phi$, $m\in \{-n,\ldots, n\}$}
\end{enumerate}
and it is equal to {$F_{j,m/2}$ (in case of $m=sn$ it is just $\Psi_j^{\mathbf{n}_s}$)}. 
\end{Lemma}

\begin{proof}
Checking decomposition.
\end{proof}

{In particular from uniqueness of the eigenstates and \eqref{eq:beta-interntwiner} we have
\begin{equation}
    \hat{\beta}(F_{j,sm/2}^{s})\propto F_{j,-sm/2}^{-s},\quad 
    \hat{\beta}(F_{j,m/2})\propto F_{j,-m/2}
\end{equation}
where $\propto$ means proportionality.
This means that basis states not only are bounded, but also they are mapped to bounded states under $\hat{\beta}$.
}

{
We will now determine proportionality constants in case of coherent states.
}

A coherent states are uniquely up to a complex number determined by conditions of being eigenfunctions of $L^2$ and $L_3$. Thus
\be
\hat{\beta}(\Psi_j^{+,\mathbf{n}_+})=C_+ \Psi_j^{-,\mathbf{n}_-},\qquad \hat{\beta}(\Psi_j^{\mathbf{n}_+})=C_0\Psi_j^{\mathbf{n}_-},
\ee
Let us now introduce
\be
U=\left(\begin{array}{cc}
0 & -1 \\
1 & 0
\end{array}\right)\in \SLDC
\ee
We check by direct inspection
\be
&U\act \Psi^{+,\mathbf{n}_+}=\Psi^{-,\mathbf{n}_-},\quad U\act \Psi^{-,\mathbf{n}_-}=(-1)^n\Psi^{+,\mathbf{n}_+},\\
&U\act \Psi^{0,\mathbf{n}_+}=\Psi^{0,\mathbf{n}_-},\quad U\act \Psi^{0,\mathbf{n}_-}=(-1)^n\Psi^{0,\mathbf{n}_+},
\ee
From intertwining property of $\hat{\beta}$
\be
\hat{\beta}(\Psi_j^{-,\mathbf{n}_-})=U\act \hat{\beta}(\Psi_j^{+,\mathbf{n}_+})=C_+U\act \Psi_j^{-,\mathbf{n}_-}=(-1)^nC_+\Psi_j^{+,\mathbf{n}_+},\qquad C_-=(-1)^nC_+
\ee
and similarly
\be
\hat{\beta}(\Psi_j^{0,\mathbf{n}_-})=(-1)^nC_0\Psi_j^{0,\mathbf{n}_+},
\ee
Using intertwining property of the action of $\hat{\beta}$ we have
\be
\hat{\beta}(\Psi_j^{\pm,\mathbf{m}^{\pm}})=C_\pm\Psi_j^{\mp,\mathbf{m}^{\mp}},\quad \hat{\beta}(\Psi_j^{0,\mathbf{m}^0})=C_0\Psi_j^{0,\mathbf{m}^0}
\ee

\subsection{Determination of \texorpdfstring{$C_\pm, C_0$}{Cpm, C0}}

We parametrize part of $\CP^1$ with $\langle \zz\mid \zz\rangle>0$ by
\be
\zz=\left(\begin{array}{l}
\cosh r\\
e^{-i\phi}\sinh r
\end{array}\right)
\ee
then
\begin{align}
&\Omega_{\zz}=\cosh r\sinh r\ dr\wedge d\phi\\
&[\zz,\mathbf{n}_-]=\cosh r,\quad \langle \zz\mid\zz\rangle=1,\quad \langle \zz\mid\mathbf{n}_+\rangle=\cosh r
\end{align}
thus
\be
\overline{\hat{\beta}(\Psi^{+\mathbf{n}_+})(\mathbf{n}_-)}=\frac{\sqrt{\rho^2+n^2}}{2\pi}
\int_0^{2\pi}d\phi\int_0^\infty dr\ \sinh r\cosh r^{-i\rho-1}\ \Phi(\zz),
\ee
Notice
\be
\Phi(\zz)=\sqrt{\frac{2j-1}{\pi}}(\cosh r)^{-n}
\ee
and finally
{
\begin{align}
&\overline{\hat{\beta}(\Psi^{+\mathbf{n}_+})(\mathbf{n}_-)}=\frac{\sqrt{\rho^2+n^2}}{2\pi}
\int_0^{2\pi}d\phi\int_0^\infty dr\ \sinh r\cosh r^{-i\rho-1} \sqrt{\frac{2j-1}{\pi}}(\cosh r)^{-n}=\\
&=\sqrt{\frac{2j-1}{\pi}}\sqrt{\rho^2+n^2}\int_1^\infty dx\ x^{-i\rho-n-1}=
\sqrt{\frac{2j-1}{\pi}}\frac{\sqrt{\rho^2+n^2}}{i\rho+n}
\end{align}
}
We can check that
\be
\Psi^{-,\mathbf{n}_-}(\mathbf{n}_-)=\sqrt{\frac{2j-1}{\pi}}
\ee
thus from $\overline{\hat{\beta}(\Psi^{+\mathbf{n}_+})(\mathbf{n}_-)}=\overline{C_+\Psi^{-,\mathbf{n}_-}(\mathbf{n}_-)}$ we get
{
\be
C_+=\frac{\sqrt{\rho^2+n^2}}{n-i\rho},\quad |C_+|=1,\quad \rho=2\g j,\quad n=2j.
\ee
}
The constant $C_0$ can be determined to be
\be
C_0=\frac{\sqrt{\rho^2+n^2}}{i\rho+n}
\ee
by similar computation and it was done in \cite{semiclassical} (equation 6.) 

\section{Classical tetrahedron and its volume}\label{Classical tetrahedron}
The shape of Lorentzian timelike tetrahedron can be parametrised by the following real variables:
\begin{eqnarray}
  A_1, A_2, A_3, A_4, \theta_{12}, \theta_{13}
\end{eqnarray}
where $A_i$ are areas asscoiated to $i$-th face. $\theta_{12}$ and $\theta_{13}$ are dihedral angles between faces $1,2$ and $1,3$:
\begin{eqnarray}
  \cosh(\theta_{12}) = s_{12}   (\vec{n}_{1} \cdot \vec{n}_{2}) \,, \\
   \cosh(\theta_{13}) = s_{13}  (\vec{n}_{1} \cdot \vec{n}_{3}) \,,
\end{eqnarray}
with $s_{12} := \sgn(\vec{n}_{1} \cdot \vec{n}_{2}) $.
The closure condition for $4$ normal vectors 
\begin{eqnarray}
  0 = \sum_{i=1}^4 A_{i} \vec{n}_{i}
\end{eqnarray}
Here we suppose all of the normals $n_i$ are timelike.
When $s_{12}>0$,
\begin{eqnarray}
  -|A_{1} \vec{n}_{1} + A_{2} \vec{n}_{2}|^2 = A_{1}^2 + A_{2}^2 - 2 A_{1}A_{2} \cosh(\theta_{12}) \leq A_{1}^2 + A_{2}^2 + 2  A_{1}A_{2} = |A_1 - A_2|^2 \,,
\end{eqnarray}
while $s_{12}< 0$,
\begin{eqnarray}
   -|A_{1} \vec{n}_{1} + A_{2} \vec{n}_{2}|^2 = A_{1}^2 + A_{2}^2 + 2 A_{1}A_{2} \cosh(\theta_{12}) \geq A_{1}^2 + A_{2}^2 + 2 A_{1}A_{2} = |A_1 + A_2|^2
\end{eqnarray}
The volume of the tetrahedron is given by
\begin{eqnarray}
V^4 &=& \frac{1}{3^4} {{A_1}}^2 \Big[-A_1^4-2 A_2^2 \cosh (2 \theta_{12}) \left(A_1^2+A_3^2\right)-2 A_3^2 \cosh (2 \theta_{13}) \left(A_1^2+A_2^2\right)-\left(A_2^2+A_3^2-A_4^2\right)^2 \nonumber\\
&&-4 s_{12}s_{13} A_2 A_3 \cosh (\theta_{12}) \cosh (\theta_{13}) \left(3 A_1^2+A_2^2+A_3^2-A_4^2\right)+2 A_1^2 \left(A_4^2-2 \left(A_2^2+A_3^2\right)\right) \Big] \nonumber \\
&& +  \frac{4 s_{12}}{3^4} {{A_1}}^3 \big(A_2 \cosh (\theta_{12})+s_{12} s_{13} A_3 \cosh (\theta_{13})\big)\\
&&\times \Big[A_1^2+A_2^2+A_3^2-A_4^2+2 s_{12} s_{13} A_2 A_3 \cosh (\theta_{12}) \cosh (\theta_{13}) \Big] \nonumber
\end{eqnarray}
In case of $A_1=A_2=A_3=A_4=1$ and $s_{12}=-s_{13}=-1$
\begin{eqnarray}
V=\frac{2}{3} \sqrt[4]{2} \sqrt[4]{\cosh ^2\left(\frac{\theta_{12}}{2}\right) \sinh ^2\left(\frac{\theta_{13}}{2}\right) (\cosh (\theta_{12})-\cosh (\theta_{13}))}
\end{eqnarray}
When $V$ is fixed, there is still a noncompact space of $\theta_{12},\theta_{13}$ since $\theta_{12},\theta_{13}$ can be very large provided $|\theta_{12}| - |\theta_{13}|$ is very small. 

\begin{Lemma}

When $s_{12}=s_{13}$, i.e. there are 1 face-normal future-pointing (past-pointing) and 3 face-normals past-pointing (future-pointing), the space of $\theta_{12},\theta_{13}$ is compact.  

\end{Lemma}

\emph{Proof.} Suppose the $4$ normal vectors are given by $u_i = A_i n_i = (u_i^0,\vec{u}_i), i =1,\dots,4$, $- u_i \cdot u_i = (u_i^0)^2 - |\vec{u}_i|^2 =A_i^2$, and we choose $u_4$ to be past pointing while the others are future pointing.
By rescaling and rotating we can fix $u_1 = (1,0,0,0)$, then
\be 
|u_1 \cdot u_2| = | u_2^0 | < | u_4^0 |\,, \qquad
|u_1 \cdot u_3| = | u_3^0 | < | u_4^0 |
\ee 
Thus we only need to proof that $| u_4^0 |$ is bounded. By writing $u_2 + u_3 = u_{23}$, one have $-|u_{23}|^2 =(u_2^0+u_3^0)^2 - (\vec{u}_2 +\vec{u}_3 )^2 > (u_2^0+u_3^0)^2 - (|\vec{u}_2| + |\vec{u}_3| )^2 > 0$, which implies $u_{23}$ must be timelike since $u_2,u_3$ timelike. 
The closure condition implies
\be
|u_4^0| - 1 =u_{23}^0 \, , \qquad -\vec{u}_4 = \vec{u}_{23} 
\ee
with $|u_4^0| >1$ since $u_2,u_3$ future pointing thus $u_{23}^0 >0$.
Suppose $-| u_{23}|^2 = A_{23}^2 > 0|$, using $| u_4^0 |^2 = A_4^2 + |\vec{u}_4|^2, | u_{23}^0 |^2 = A_{23}^2 + |\vec{u}_{23}|^2$, we have
\be 
| |u_{4}^0| -1 |^2 =|u_4^0|^2 + 1 - 2 |u_{4}^0| = A_{23}^2 + |\vec{u}_{4}|^2 = | u_4^0 |^2 +  A_{23}^2 - A_4^2
\ee 
which implies
\be 
 1 - 2 |u_{4}^0| =  A_{23}^2 - A_4^2
\ee 
As a result, $0 < A_{23}^2 < A_4^2$ and $|u_{4}^0| \leq \frac{A_4^2 + 1}{2}$ which is bounded.\\ 
\noindent
$\Box$

\bibliographystyle{jhep}

\bibliography{muxin}

\providecommand{\href}[2]{#2}\begingroup\raggedright\begin{thebibliography}{10}

\bibitem{book1}
C.~Rovelli, \emph{Quantum Gravity}, Cambridge University Press (2004).

\bibitem{book}
T.~Thiemann, \emph{Modern Canonical Quantum General Relativity}, Cambridge
  University Press (2007).

\bibitem{rovelli2014covariant}
C.~Rovelli and F.~Vidotto, \emph{Covariant Loop Quantum Gravity: An Elementary
  Introduction to Quantum Gravity and Spinfoam Theory}, Cambridge Monographs on
  Mathematical Physics, Cambridge University Press (2014).

\bibitem{Perez2012}
A.~Perez, \emph{{The Spin Foam Approach to Quantum Gravity}},
  \href{https://doi.org/10.12942/lrr-2013-3}{\emph{Living Rev.Rel.} {\bfseries
  16} (2013) 3} [\href{https://arxiv.org/abs/1205.2019}{{\ttfamily
  1205.2019}}].

\bibitem{EPRL}
J.~Engle, E.~Livine, R.~Pereira and C.~Rovelli, \emph{{LQG vertex with finite
  Immirzi parameter}},
  \href{https://doi.org/10.1016/j.nuclphysb.2008.02.018}{\emph{Nucl.Phys.}
  {\bfseries B799} (2008) 136}
  [\href{https://arxiv.org/abs/0711.0146}{{\ttfamily 0711.0146}}].

\bibitem{Conrady:2010kc}
F.~Conrady and J.~Hnybida, \emph{{A spin foam model for general Lorentzian
  4-geometries}},
  \href{https://doi.org/10.1088/0264-9381/27/18/185011}{\emph{Class. Quant.
  Grav.} {\bfseries 27} (2010) 185011}
  [\href{https://arxiv.org/abs/1002.1959}{{\ttfamily 1002.1959}}].

\bibitem{KKL}
W.~Kaminski, M.~Kisielowski and J.~Lewandowski, \emph{{Spin-Foams for All Loop
  Quantum Gravity}}, \href{https://doi.org/10.1088/0264-9381/29/4/049502,
  10.1088/0264-9381/27/9/095006}{\emph{Class. Quant. Grav.} {\bfseries 27}
  (2010) 095006} [\href{https://arxiv.org/abs/0909.0939}{{\ttfamily
  0909.0939}}].

\bibitem{generalize}
Y.~Ding, M.~Han and C.~Rovelli, \emph{{Generalized Spinfoams}},
  \href{https://doi.org/10.1103/PhysRevD.83.124020}{\emph{Phys.Rev.} {\bfseries
  D83} (2011) 124020} [\href{https://arxiv.org/abs/1011.2149}{{\ttfamily
  1011.2149}}].

\bibitem{Kaminski:2017eew}
W.~Kaminski, M.~Kisielowski and H.~Sahlmann, \emph{{Asymptotic analysis of the
  EPRL model with timelike tetrahedra}},
  \href{https://doi.org/10.1088/1361-6382/aac6a4}{\emph{Class. Quant. Grav.}
  {\bfseries 35} (2018) 135012}
  [\href{https://arxiv.org/abs/1705.02862}{{\ttfamily 1705.02862}}].

\bibitem{Liu:2018gfc}
H.~Liu and M.~Han, \emph{{Asymptotic analysis of spin foam amplitude with
  timelike triangles}},
  \href{https://doi.org/10.1103/PhysRevD.99.084040}{\emph{Phys. Rev.}
  {\bfseries D99} (2019) 084040}
  [\href{https://arxiv.org/abs/1810.09042}{{\ttfamily 1810.09042}}].

\bibitem{Simao:2021qno}
J.D.~Sim\~ao and S.~Steinhaus, \emph{{Asymptotic analysis of spin-foams with
  time-like faces in a new parameterisation}},
  \href{https://arxiv.org/abs/2106.15635}{{\ttfamily 2106.15635}}.

\bibitem{Baez:2001fh}
J.C.~Baez and J.W.~Barrett, \emph{{Integrability for relativistic spin
  networks}}, \href{https://doi.org/10.1088/0264-9381/18/21/316}{\emph{Class.
  Quant. Grav.} {\bfseries 18} (2001) 4683}
  [\href{https://arxiv.org/abs/gr-qc/0101107}{{\ttfamily gr-qc/0101107}}].

\bibitem{Engle:2008ev}
J.~Engle and R.~Pereira, \emph{{Regularization and finiteness of the Lorentzian
  LQG vertices}}, \href{https://doi.org/10.1103/PhysRevD.79.084034}{\emph{Phys.
  Rev. D} {\bfseries 79} (2009) 084034}
  [\href{https://arxiv.org/abs/0805.4696}{{\ttfamily 0805.4696}}].

\bibitem{Kaminski:2010qb}
W.~Kaminski, \emph{{All 3-edge-connected relativistic BC and EPRL spin-networks
  are integrable}},  \href{https://arxiv.org/abs/1010.5384}{{\ttfamily
  1010.5384}}.

\bibitem{NP}
K.~Noui and P.~Roche, \emph{{Cosmological deformation of Lorentzian spin foam
  models}},
  \href{https://doi.org/10.1088/0264-9381/20/14/318}{\emph{Class.Quant.Grav.}
  {\bfseries 20} (2003) 3175}
  [\href{https://arxiv.org/abs/gr-qc/0211109}{{\ttfamily gr-qc/0211109}}].

\bibitem{QSF}
M.~Han, \emph{{4-dimensional spin-foam model with quantum Lorentz group}},
  \href{https://doi.org/10.1063/1.3606592}{\emph{J.Math.Phys.} {\bfseries 52}
  (2011) 072501} [\href{https://arxiv.org/abs/1012.4216}{{\ttfamily
  1012.4216}}].

\bibitem{QSF1}
W.J.~Fairbairn and C.~Meusburger, \emph{{Quantum deformation of two
  four-dimensional spin foam models}},
  \href{https://doi.org/10.1063/1.3675898}{\emph{J.Math.Phys.} {\bfseries 53}
  (2012) 022501} [\href{https://arxiv.org/abs/1012.4784}{{\ttfamily
  1012.4784}}].

\bibitem{Han:2021tzw}
M.~Han, \emph{{Four-dimensional Spinfoam Quantum Gravity with Cosmological
  Constant: Finiteness and Semiclassical Limit}},
  \href{https://arxiv.org/abs/2109.00034}{{\ttfamily 2109.00034}}.

\bibitem{gelfand5}
I.M.~{Gel'fand}, M.I.~{Graev} and N.Y.~{Vilenkin}, \emph{Generalized Functions
  V. Integral Geometry and Representation Theory}, Academic Press, 111 Fifth
  Avenue, New York, New York 10003 (1966).

\bibitem{semiclassical}
J.W.~Barrett, R.~Dowdall, W.J.~Fairbairn, F.~Hellmann and R.~Pereira,
  \emph{{Lorentzian spin foam amplitudes: Graphical calculus and asymptotics}},
  \href{https://doi.org/10.1088/0264-9381/27/16/165009}{\emph{Class.Quant.Grav.}
  {\bfseries 27} (2010) 165009}
  [\href{https://arxiv.org/abs/0907.2440}{{\ttfamily 0907.2440}}].

\bibitem{Bahr:2017klw}
B.~Bahr and S.~Steinhaus, \emph{{Hypercuboidal renormalization in spin foam
  quantum gravity}},
  \href{https://doi.org/10.1103/PhysRevD.95.126006}{\emph{Phys. Rev. D}
  {\bfseries 95} (2017) 126006}
  [\href{https://arxiv.org/abs/1701.02311}{{\ttfamily 1701.02311}}].

\bibitem{Bianchi:2010zs}
E.~Bianchi, C.~Rovelli and F.~Vidotto, \emph{{Towards Spinfoam Cosmology}},
  \href{https://doi.org/10.1103/PhysRevD.82.084035}{\emph{Phys. Rev. D}
  {\bfseries 82} (2010) 084035}
  [\href{https://arxiv.org/abs/1003.3483}{{\ttfamily 1003.3483}}].

\bibitem{Dona:2020yao}
P.~Dona and S.~Speziale, \emph{{Asymptotics of lowest unitary SL(2,C)
  invariants on graphs}},
  \href{https://doi.org/10.1103/PhysRevD.102.086016}{\emph{Phys. Rev. D}
  {\bfseries 102} (2020) 086016}
  [\href{https://arxiv.org/abs/2007.09089}{{\ttfamily 2007.09089}}].

\bibitem{BC}
J.W.~Barrett and L.~Crane, \emph{{Relativistic spin networks and quantum
  gravity}}, \href{https://doi.org/10.1063/1.532254}{\emph{J.Math.Phys.}
  {\bfseries 39} (1998) 3296}
  [\href{https://arxiv.org/abs/gr-qc/9709028}{{\ttfamily gr-qc/9709028}}].

\bibitem{Rovelli:2010vv}
C.~Rovelli, \emph{{Simple model for quantum general relativity from loop
  quantum gravity}},
  \href{https://doi.org/10.1088/1742-6596/314/1/012006}{\emph{J. Phys. Conf.
  Ser.} {\bfseries 314} (2011) 012006}
  [\href{https://arxiv.org/abs/1010.1939}{{\ttfamily 1010.1939}}].

\bibitem{HZ}
M.~Han and M.~Zhang, \emph{{Asymptotics of spinfoam amplitude on simplicial
  manifold: Lorentzian theory}},
  \href{https://doi.org/10.1088/0264-9381/30/16/165012}{\emph{Class.Quant.Grav.}
  {\bfseries 30} (2013) 165012}
  [\href{https://arxiv.org/abs/1109.0499}{{\ttfamily 1109.0499}}].

\bibitem{hanPI}
M.~Han and T.~Krajewski, \emph{{Path integral representation of Lorentzian
  spinfoam model, asymptotics, and simplicial geometries}},
  \href{https://doi.org/10.1088/0264-9381/31/1/015009}{\emph{Class.Quant.Grav.}
  {\bfseries 31} (2014) 015009}
  [\href{https://arxiv.org/abs/1304.5626}{{\ttfamily 1304.5626}}].

\bibitem{Bahr:2010bs}
B.~Bahr, F.~Hellmann, W.~Kaminski, M.~Kisielowski and J.~Lewandowski,
  \emph{{Operator Spin Foam Models}},
  \href{https://doi.org/10.1088/0264-9381/28/10/105003}{\emph{Class. Quant.
  Grav.} {\bfseries 28} (2011) 105003}
  [\href{https://arxiv.org/abs/1010.4787}{{\ttfamily 1010.4787}}].

\bibitem{Han:2020npv}
M.~Han, Z.~Huang, H.~Liu, D.~Qu and Y.~Wan, \emph{{Spinfoam on a Lefschetz
  thimble: Markov chain Monte Carlo computation of a Lorentzian spinfoam
  propagator}}, \href{https://doi.org/10.1103/PhysRevD.103.084026}{\emph{Phys.
  Rev. D} {\bfseries 103} (2021) 084026}
  [\href{https://arxiv.org/abs/2012.11515}{{\ttfamily 2012.11515}}].

\bibitem{Dona:2019dkf}
P.~Dona, M.~Fanizza, G.~Sarno and S.~Speziale, \emph{{Numerical study of the
  Lorentzian Engle-Pereira-Rovelli-Livine spin foam amplitude}},
  \href{https://doi.org/10.1103/PhysRevD.100.106003}{\emph{Phys. Rev.}
  {\bfseries D100} (2019) 106003}
  [\href{https://arxiv.org/abs/1903.12624}{{\ttfamily 1903.12624}}].

\bibitem{github}
H.~Liu.
  \url{https://github.com/LQG-Florida-Atlantic-University/extended_spinfoam},
  2021.

\bibitem{HHKR}
H.M.~Haggard, M.~Han, W.~Kaminski and A.~Riello, \emph{{SL(2,C) Chern-Simons
  Theory, a non-Planar Graph Operator, and 4D Loop Quantum Gravity with a
  Cosmological Constant: Semiclassical Geometry}},
  \href{https://doi.org/10.1016/j.nuclphysb.2015.08.023}{\emph{Nucl. Phys.}
  {\bfseries B900} (2015) 1} [\href{https://arxiv.org/abs/1412.7546}{{\ttfamily
  1412.7546}}].

\end{thebibliography}\endgroup
\end{document}